\documentclass{article}
\usepackage[dvips]{graphicx}
\usepackage{amsmath}
\usepackage{subfigure}
\usepackage{amssymb,amsthm}
\usepackage{tikz}
\usepackage{newalg}
\usepackage{blkarray}
\usepackage{cite}%\usepackage{hyperref}

\newcommand{\url}[1]{#1}
\DeclareMathOperator{\F}{\mathbb F}
\DeclareMathOperator{\wt}{wt}

\title{Polar Subcodes}
\author{Peter Trifonov
 and Vera Miloslavskaya
\thanks{The authors are with the Distributed Computing and Networking 
Department, Saint-Petersburg Polytechnic University, 
Polytechnicheskaya str., 21, office 103, 194021, Saint-Petersburg, Russia,
Email: petert@dcn.icc.spbstu.ru. }
\thanks{This work was partially presented at IEEE Information Theory Workshop
2013 and International Symposium on Information Theory and Its Applications
2014.}
\thanks{This work was partially supported by the Russian Foundation for Basic Research under the grant 12-01-00365-a and the Samsung Global Research Outreach
grant.}
}

\date{\today}
\theoremstyle{plain}

\newtheorem{theorem}{Theorem}
\newtheorem{definition}{Definition}
\newtheorem{example}{Example}

\newcommand{\mF}{\mathcal F}
\renewcommand{\F}{\mathbb F}
\newcommand{\set}[1]{\left\{{#1}\right\}}

\newcommand{\floor}[1]{\left\lfloor{#1}\right\rfloor}

\DeclareMathOperator{\Tr}{Tr}

\begin{document}

        \maketitle.
\begin{abstract}
An extension of polar codes is proposed, which allows some of the frozen
symbols, called dynamic frozen symbols,  to be data-dependent. 
A construction of  polar codes with dynamic frozen symbols, being subcodes of extended BCH codes, is proposed. The proposed codes have higher minimum distance than classical polar codes, but  
still can be efficiently decoded 
using the successive cancellation  algorithm and its  
extensions. The codes with Arikan,  extended BCH and Reed-Solomon kernel are considered.
The  proposed codes are shown to outperform LDPC and turbo codes, as well as polar codes with CRC.
\end{abstract}
%\begin{IEEEkeywords}
%Polar codes, generalized concatenated codes, BCH\ codes, Reed-Solomon codes.
%\end{IEEEkeywords}

\section{Introduction}
Polar codes were recently shown to be able to achieve the capacity of binary
input memoryless output-symmetric channels \cite{arikan2009channel}. Low-complexity construction,
encoding and decoding algorithms are available for polar codes. However, the performance of polar
codes  of moderate length appears to be quite poor.  This is both due to
suboptimality of the successive cancellation (SC)\ decoding algorithm and
low minimum distance of polar codes. The first problem can be solved by employing
list/stack SC decoding techniques \cite{tal2011list,niu2012stack,chen2013improved,miloslavskaya2014sequential,trifonov2013polar,miloslavskaya2014sequentialBCH},
which far outperform  the SC algorithm. Alternatively, one can use
the belief propagation algorithm   \cite{guo2014enhanced}. Its performance,
however, is still inferior to list/stack SC decoding.
 
The second problem  can be solved by constructing a generalized concatenated
code with inner polar codes \cite{trifonov2012efficient,mahdavifar2013construction,wang2014concatenations},
 or employing a serial concatenation of an error detecting or error correcting code and
a polar code \cite{bakshi2010concatenated,tal2011list,li2012adaptive,niu2012crcaided,guo2014enhanced}. However,
in the second case it is not clear how the parameters of the outer codes affect the minimum distance
and finite-length performance of the concatenated code.

It was shown recently that a sequence of linear codes achieves capacity on
a memoryless erasure channel under MAP decoding if their
blocklengths are strictly increasing, rates converge
to some $r\in (0, 1)$, and the permutation group of each code
is doubly transitive  \cite{kumar2015reed,kudekar2015reed}. This class of
codes includes Reed-Muller (RM) and extended primitive narrow-sense BCH (EBCH) codes. Observe that RM
codes can be considered as a special case of polar codes. On the other hand,
EBCH codes are known to have much higher minimum distance than comparable
RM codes, and are therefore likely to provide better finite length performance. However, there are still no efficient MAP decoding algorithms for these codes.

It was suggested  in \cite{dumer2006list} to construct subcodes of RM
codes, which can be efficiently decoded by a recursive list decoding algorithm.
In this paper we generalize this approach, and propose a code construction
"in between" polar codes and EBCH\ codes. The proposed codes can be efficiently decoded using the techniques developed in the area of polar coding, but provide much higher minimum distance, which can be accurately controlled.  The obtained codes outperform  state-of-the art LDPC, turbo and polar codes. More specifically,
in   Section \ref{sIGCC}
we introduce an extension of generalized concatenated codes (GCC), called
interlinked generalized concatenated codes (IGCC). Recursive application of this construction enables one to represent
a linear block code in a form which, in principle, enables its decoding
by the SC algorithm (Section \ref{sDFS}). This form, called polar codes with
dynamic frozen symbols, can be considered as a generalization of polar codes.
We show that EBCH\ codes are particularly well suited
for such representation, although their SC decoding is still not very efficient.
Furthermore, we present a special case of IGCC, called polar
subcodes, with good performance
under the SC algorithm and its derivatives (Section \ref{sPolarSubcodes}). The proposed codes are subcodes of
EBCH codes. We consider polar subcodes with
Arikan, EBCH\ and Reed-Solomon kernel. Simulation results presented in Section \ref{sNumeric} show that the proposed codes
outperform state-of-the-art polar, LDPC and turbo codes.

\section{Background}
\label{sBackground}
\subsection{Generalized concatenated codes}

A generalized concatenated code (GCC) \cite{blokh1974coding} over $\F_q$ is defined using a family of nested inner $(n,k_i,d_i)$ codes $\mathcal C_i: \mathcal C_0\supset \mathcal C_1\supset \dots \supset \mathcal C_{\nu-1}$, and a family of outer $(N,K_i,D_i)$ codes $\mathbb C_i$, where the $i$-th outer code is defined over $\F_{q^{k_i-k_{i+1}}}, 0\leq i<\nu,k_\nu=0$. It will be assumed in this paper that $k_i=k_{i+1}+1, \nu=n$. Let $\mathcal G$ be a $n\times n$ matrix, such that its rows $i,\dots,n-1$  generate code $\mathcal C_i$. GCC encoding is performed as follows. First, partition a data vector into $n$ blocks of size $K_i, 0\leq i<n$. Second,  encode these blocks with codes $\mathbb C_i$ to obtain codewords $(\tilde c_{i,0},\dots,\tilde c_{i,N-1})$. Finally, multiply vectors $(\tilde c_{0,j},\dots,\tilde c_{n-1,j}), 0\leq j<N,$  by $\mathcal G$ to obtain a GCC codeword $(c_{0,0},\dots,c_{n-1,0},c_{0,1},\dots,c_{n-1,N-1})$. Figure \ref{fGCC} illustrates this construction.  A GCC generator matrix can be obtained as $$G=\begin{pmatrix}G^{(0)}\otimes \mathcal G_{0,-}\\G^{(1)}\otimes \mathcal G_{1,-}\\\vdots\\G^{(n-1)}\otimes \mathcal G_{n-1,-}\end{pmatrix},$$
where $G^{(i)}$ is a generator matrix of $\mathbb C_i$, and $\mathcal G_{i,-}$ denotes the $i$-th row of $\mathcal G$.  It is possible to show that this encoding method results in a $(Nn,\sum_{i=0}^{n-1}K_i,\geq \min_i d_iD_i)$ linear block code.
\begin{figure}
\includegraphics[width=0.45\textwidth]{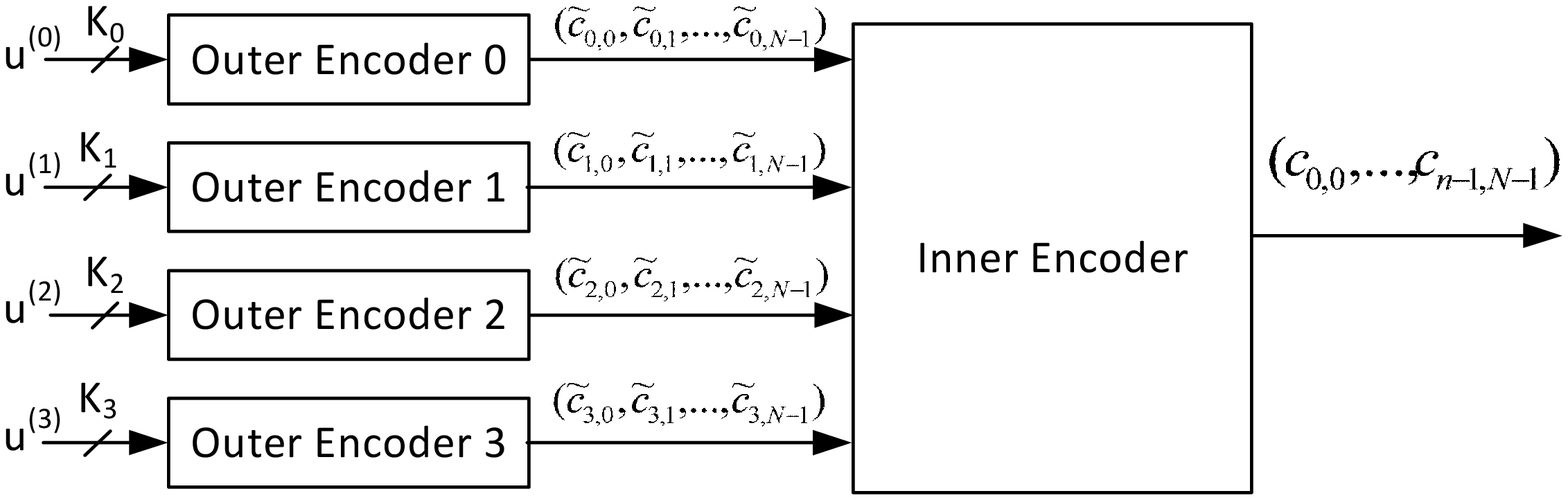}
\caption{Encoding with a generalized concatenated code   }
\label{fGCC}
\end{figure}

\begin{figure}
\includegraphics[width=0.45\textwidth]{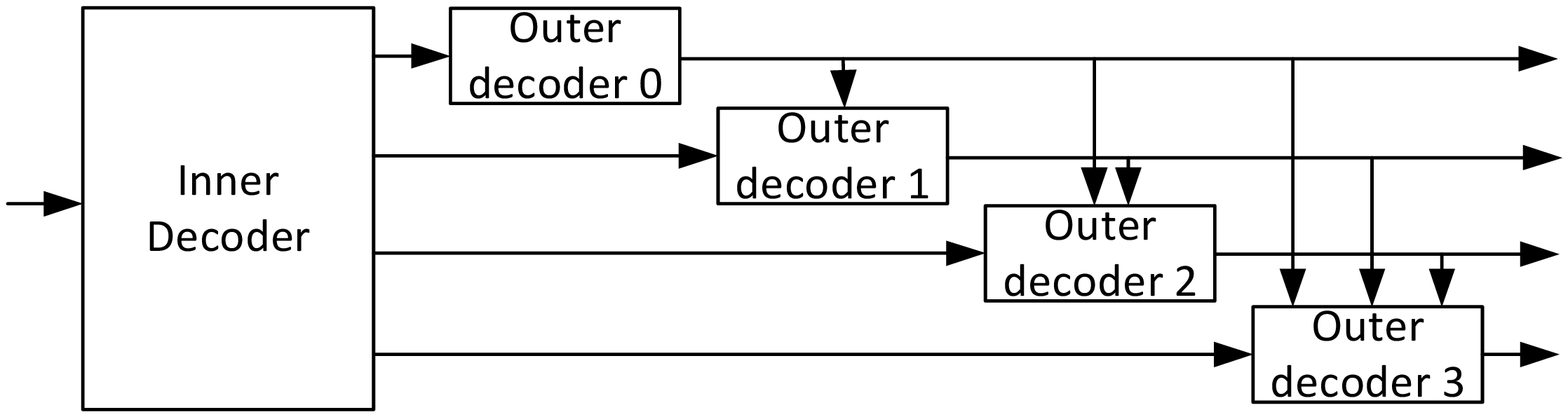}
\caption{Multistage decoding     }
\label{fMSD}
\end{figure}
    
GCC can be decoded with a multistage decoding (MSD) algorithm \cite{imai77new,takata1994suboptimum,wachsmann99multilevel}. For $i=0,1,\dots,N-1$, this algorithm takes as input noisy instances $y_{t,j}$ of codeword symbols $c_{t,j}, 0\leq t<n, 0\leq j<N$,  successively computes estimates of $\tilde c_{i,j}$ using a SISO decoder of $\mathcal C_i$, and passes these estimates
to a decoder of $\mathbb C_i$ to recover the corresponding codeword. Then it proceeds with decoding of $\mathcal C_{i+1}$ and $\mathbb C_{i+1}$, as shown in Figure \ref{fMSD}. The performance of the MSD algorithm depends strongly on parameters of outer codes. An extensive survey of various methods for their selection can be found in \cite{wachsmann99multilevel}.

\subsection{Polar codes}
$(n=l^m,k)$ polar code over $\F_q$ is a linear block code generated by $k$ rows of 
matrix $A=B_{l,m}F_{l}^{\otimes m}$, where $B_{l,m}$ is the digit-reversal permutation matrix, $F_l$ is a $l\times l$ matrix called kernel (e.g. $F_2=\begin{pmatrix}1&0\\1&1\end{pmatrix}$ is the Arikan kernel), and $\otimes m$ denotes 
$m$-times Kronecker product of the matrix with itself \cite{arikan2009channel}.  The digit-reversal permutation maps integer
  $i=\sum_{j=0}^{m-1}i_jl^j, 0\leq i_j<l,$ onto $\sum_{j=0}^{m-1}i_jl^{m-1-j}$.
  The particular rows 
to be used in a generator matrix are selected so that the  error 
probability under the below described successive cancellation (SC)\ decoding algorithm is minimized. Hence, a codeword of a classical polar code is obtained as 
$c=uA$, where $u_i=0, i\in \mathcal F,$ and $\mathcal F\subset\set{0,\ldots,n-1}$ is the 
set of $n-k$ frozen symbol indices.  It is possible to show that matrix $A$
transforms the original binary input memoryless output-symmetric channel
$W_1^{(0)}(y|c)$ into bit subchannels  $W_n^{(i)}(y_0^{n-1},u_0^{i-1}|u_i)$,
the capacities of these subchannels converge with $m$ to $0$ or $1$ symbols
per channel use, and the fraction
of subchannels with capacity close to $1$ converges to the capacity of $W_1^{(0)}(y|c)$.
Here $a_s^t=(a_s,\ldots,a_t)$, and $y_0,\ldots,y_{n-1}$ are the noisy symbols obtained by transmitting codeword symbols $c_0,\ldots,c_{n-1}$ over a 
binary input memoryless output-symmetric channel $W_1^{(0)}(y|u)$.

The SC decoding algorithm at phase $i$
computes 
$W_{n}^{(i)}(y_0^{n-1},u_0^{i-1}|u_i), 
u_i\in \F_q$ (or $W_n^{(i)}(u_0^i|y_0^{n-1})$, which is
more convenient for implementation), and  makes 
decisions 
\begin{equation}
\label{mSCDecisionRule}
\widehat{u}_i=\begin{cases} \arg \max_{u_i\in\F_q} W_{n}^{(i)}(u_0^{i}|y_0^{n-1}),& i\not \in \mathcal F \\0,&\text{otherwise}.\end{cases}
\end{equation}
 This decision 
is used at subsequent steps instead of the true value of $u_i$ to determine the values of 
$u_{i+1},\ldots,u_{n-1}$.  It was shown in \cite{arikan2009channel} that 
these calculations can be implemented with complexity $O(n\log n)$.  For
example, in the case of $l=2,q=2$ these probabilities can be computed as 
\begin{align}
\label{mSCProb1}
\begin{split}
\lefteqn{W_n^{(2i)}({u}_0^{2i}|y_0^{n-1})=}&\\
&\displaystyle\sum_{u_{2i+1}=0}^1 W_{\frac{n}{2}}^{(i)}
({u}_{0,even}^{2i+1} \oplus { u}_{0,odd}^{2i+1}|y_0^{\frac{n}{2}-1})W_{\frac{n}{2}}^{(i)}( u_{0,odd}^{2i+1} |y_\frac{n}{2}^{n-1})
\end{split}
\\
\label{mSCProb2}
\begin{split}
\lefteqn{W_n^{(2i+1)}({u}_0^{2i+1}|y_0^{n-1})=}&\\
& W_{\frac{n}{2}}^{(i)}
({u}_{0,even}^{2i+1} \oplus { u}_{0,odd}^{2i+1}|y_0^{\frac{n}{2}-1})W_{\frac{n}{2}}^{(i)}( u_{0,odd}^{2i+1} |y_\frac{n}{2}^{n-1}).
\end{split}
\end{align}

For $q=2$ the Bhattacharyya parameters $Z_{n,i}$ of the bit subchannels  $W_n^{(i)}(y_0^{n-1},u_0^{i-1}|u_i)$ satisfy \cite{korada2010polar}
\begin{equation}
Z_{n/l,i}^{\Delta_{j}}\leq Z_{n,il+j}\leq 2^{l-{j}}Z_{n/l,i}^{\Delta_{j}},
\end{equation}
where $Z_{1,0}$ is the Bhattacharryya parameter of the original binary memoryless symmetric
channel, and $\Delta_i,0\leq i<l,$ are partial distances of matrix $F_l$.  Similar bounds are provided in \cite{mori2014source} for $q>2$. For the case of $l=2$ (Arikan kernel), one can obtain more precise estimates as \cite{arikan2009channel,hassani2010scaling}
\begin{align}
\label{mArikanBhat1}
Z_{n/2,i}\sqrt{2-Z_{n/2,i}^2}\leq Z_{n,2i}\leq& 2Z_{n/2,i}-Z_{n/2,i}^2\\
\label{mArikanBhat2}
Z_{n,2i+1}=&Z_{n/2,i}^2.
\end{align}
Furthermore, for the case of the binary erasure channel, one has $Z_{n,2i}= 2Z_{n/2,i}-Z_{n/2,i}^2$. 

Let $P_i=1-P\set{C_i|C_0,\dots,C_{i-1}}$ be the error probability of symbol $u_i$ under SC decoding, where $C_i$ is the event corresponding to correct estimation of symbol $u_i$. Then the SC decoding error probability is given by 
\begin{equation}
\label{mSCErrProb}
P=1-\prod_{i\notin \mF}(1-P_i).
\end{equation}
Efficient techniques are available for computing $P_i$ in the case of Arikan kernel \cite{tal2011how,trifonov2012efficient}. The standard way to construct practical polar codes is to select $\mF$ as the set of $n-k$ indices $i$  with the highest error probability $P_i$.
 
For any $s,0< s\leq m$, polar codes can be considered as GCC with inner codes generated by the rows of matrix $F_l^{\otimes s}$,  and outer codes generated by some submatrices of $F_l^{\otimes (m-s)}$.  The SC decoding algorithm can be considered as an instance of the MSD method, where symbol-by-symbol decoding of outer codes is used. It was shown in \cite{trifonov2012efficient} that significant performance improvement can be achieved by employing near-ML\ decoding algorithms for outer codes. Even better performance can be obtained by employing list or stack decoding algorithms \cite{tal2011list,niu2012stack,chen2013improved,trifonov2013polar,miloslavskaya2014sequential}.
These algorithms keep track of a number of vectors  $\hat u_0^{i-1}$, and at each
step increase the length of one or more of these vectors by $1$, and compute
probabilities $W_n^{(i)}(\hat u_0^{i}|y_0^{n-1})$ (or related values).  Vectors with low probabilities are discarded, so that there are
at most $L$ vectors of length $i$ for each $i$. The worst-case complexity
of these algorithms is $O(L n\log n)$.

\subsection{BCH codes}
An $(n=q^m,k,\geq d)$ extended primitive narrow-sense BCH  (EBCH) code is a set of vectors
$c_0^{n-1}\in \F_q^n$, such that $\sum_{i=0}^{n-1}c_ix_i^j=0, 0\leq j<d-1,$
where $(x_0,\dots,x_{n-1})$ is a vector of distinct values of $\F_{q^m}$,
called code locators. Setting $x_0=0,x_i=\alpha^{i-1}, 1\leq
i<n$, results in an extended cyclic code with generator polynomial $g(x)=LCM(M_1(x),\dots,M_{d-2}(x))$,
where $M_i(x)$ is a minimal polynomial of $\alpha^i$, and $\alpha$ is a primitive
element of $\F_{q^m}$. However, in this paper, unless
stated otherwise, it will be assumed that $x_i$ are arranged in the standard
digit order, where $x_i=\sum_{j=0}^{m-1}X_{i,j}\beta_j$, $i=\sum_{j=0}^{m-1}X_{i,j}q^j,
X_{i,j}\in\set{0,\dots,q-1}$, and $\beta_0,\dots,\beta_{m-1}$ is some basis of $\F_{q^m}$.
\section{Interlinked Generalized Concatenated Codes}
\label{sIGCC}
\subsection{The construction}
In this section we present an extension of the generalized concatenated codes, called interlinked GCC (IGCC). This extension can be used to represent a broad class of linear block codes. It enables one to decode such codes using the techniques developed in the area of generalized concatenated and multilevel coding. These  decoding algorithms, however, are not guaranteed to perform
well for an arbitrary IGCC. But it will be shown below how to construct IGCC with good performance under MSD (actually, SC) and its list extensions.  

\begin{figure}
\begin{center}
\includegraphics[width=0.45\textwidth]{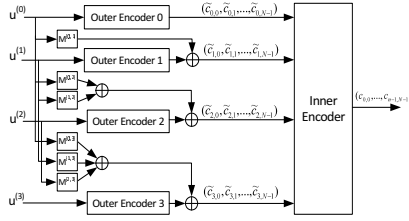}
\caption{Interlinked generalized concatenated code     }
\label{fIGCC}
\end{center}
\end{figure}

Interlinked GCC encodes the subvector $u^{(i)}\in \F_q^{K_i}$ of the data vector not with the outer code $\mathbb C_i$, as in the classical GCC, but with its  coset given by $\mathbb C_i+\left(\sum_{s=0}^{i-1} u^{(i)}M^{(s,i)}\right)$, where $M^{(s,i)}\in \F_q^{K_s\times N}$ are some matrices, as shown in Figure \ref{fIGCC}. This results in a linear block code of length $Nn$ and dimension $\sum_{i=0}^{n-1}K_i$. It is, however, quite difficult to estimate the minimum distance of the obtained code.  Obviously, for any pair of non-negative integers $\eta,k: \eta >k,$ if there exists a  $(\eta,k,d)$ GCC, then there  also exists a $(\eta,k,\geq d)$ IGCC. 

The MSD algorithm  can be used to decode IGCC.
However, one needs to perform decoding not in outer codes, but in their
cosets. This can be done with any decoder for $\mathbb
C_i$, provided that its input LLRs are appropriately adjusted.
\subsection{Generalized Plotkin decomposition of linear codes}
As a special case of IGCC, which corresponds to the case of inner codes generated
by rows of  $\mathcal G=F_2=\begin{pmatrix}1&0\\1&1\end{pmatrix}$, we consider an extension of the classical Plotkin construction. This extension will be used below to derive a generalization of Arikan polar codes.

\begin{theorem}
Any  linear $(2n,k,d)$ code $\mathcal C$ has a generator matrix given by 
\begin{equation}
\label{mGenPlotkin}
G=\begin{pmatrix}
I_{k_1}&0&\tilde I \\
0& I_{k_2}&0\\
\end{pmatrix}\begin{pmatrix}
G_1&0\\
G_2&G_2\\
G_3&G_3\\
\end{pmatrix},
\end{equation}
where $I_l$ is a $l\times l$ identity matrix,  $G_i, 1\leq i\leq 3,$   are  $k_i\times n$ matrices,  $k=k_1+k_2$,  and $\tilde I$ is obtained by stacking a $(k_1-k_3)\times k_3$ zero matrix and $I_{k_3}$ , where $k_3\leq k_1$.
\end{theorem}
\begin{proof}
%!!!Доказать, что k_1<=n!!!
Let $\tilde G=\begin{pmatrix}G'&G''\end{pmatrix}$, where $G'$ and $G''$ are some $k\times n$ matrices, be a generator matrix of the code, and let $\tilde H=\begin{pmatrix}H'&H''\end{pmatrix}$ be the corresponding parity check matrix. Let $G_2$ be a maximum rank solution of matrix equation $G_2(H'+H'')^T=0$. Gaussian elimination can be used to construct  matrix $G=Q\tilde G=\begin{pmatrix}
G_5&0\\
G_4&G_3\\
G_2&G_2\\
\end{pmatrix}$, such that $Q$ is an invertible matrix, rows of $G_3$ are linearly independent with rows of $G_2$, and $k=k_2+k_3+k_5$.
It can be seen that 
\begin{equation}
\label{mCodeDecompDetailed}
G=
\begin{pmatrix}
I_{k_5}&0&0&0\\
0&I_{k_3}&0&I_{k_3}\\
0&0&I_{k_2}&0\\
\end{pmatrix}\begin{pmatrix}
G_5&0\\
G_4-G_3&0\\
G_2&G_2\\
G_3&G_3\\
\end{pmatrix}.
\end{equation}
Then the  statement  follows by setting $G_1=\begin{pmatrix}
G_5\\
G_4-G_3
\end{pmatrix}$.
\end{proof}
Another way to construct  $G_1$ is to compute $G'+G''$, and eliminate linearly dependent rows from the obtained matrix. 

Classical Plotkin concatenation of two codes corresponds to the case of $k_3=0$,
so the representation of a generator matrix in the form \eqref{mGenPlotkin} will be referred to as a generalized Plotkin decomposition (GPD) of $G$ or the corresponding code $\mathcal C$. 
Applying the GPD to equivalent codes may result in codes $\mathcal C_1, \mathcal C_2$ with  different dimensions and performance.  
\begin{figure*}[t]
\begin{equation*}
u_0^{15}
\begin{small}
\left(
\setlength{\arraycolsep}{1pt}
\begin{array}{cccccccccccccccc}
1&0&0&0&0&0&0&0&0&0&0&0&0&0&0&0\\    
1&0&0&0&0&0&0&0&1&0&0&0&0&0&0&0\\    
1&0&0&0&1&0&0&0&0&0&0&0&0&0&0&0\\    
1&0&0&0&1&0&0&0&1&0&0&0&1&0&0&0\\    
1&0&1&0&0&0&0&0&0&0&0&0&0&0&0&0\\    
1&0&1&0&0&0&0&0&1&0&1&0&0&0&0&0\\    
1&0&1&0&1&0&1&0&0&0&0&0&0&0&0&0\\    
1&0&1&0&1&0&1&0&1&0&1&0&1&0&1&0\\    
1&1&0&0&0&0&0&0&0&0&0&0&0&0&0&0\\    
1&1&0&0&0&0&0&0&1&1&0&0&0&0&0&0\\    
1&1&0&0&1&1&0&0&0&0&0&0&0&0&0&0\\    
1&1&0&0&1&1&0&0&1&1&0&0&1&1&0&0\\    
1&1&1&1&0&0&0&0&0&0&0&0&0&0&0&0\\    
1&1&1&1&0&0&0&0&1&1&1&1&0&0&0&0\\    
1&1&1&1&1&1&1&1&0&0&0&0&0&0&0&0\\    
1&1&1&1&1&1&1&1&1&1&1&1&1&1&1&1\\    
\end{array}
\right)
\setlength{\arraycolsep}{1pt}
%\underbrace{
\begin{pmatrix}
1&0&0\\
1&1&1\\
1&\alpha&\alpha^3\\
1&1+\alpha&(1+\alpha)^3\\
1&\alpha^2&1+\alpha+\alpha^2+\alpha^3\\
1&1+\alpha^2&1+\alpha\\
1&\alpha(1+\alpha)&1+\alpha^2\\
1&1+\alpha+\alpha^2&1+\alpha+\alpha^2+\alpha^3\\
1&\alpha^3&1+\alpha^2\\
1&1+\alpha^3&1+\alpha\\
1&\alpha(1+\alpha^2)&1\\
1&1+\alpha+\alpha^3&1\\
1&\alpha^2(1+\alpha)&1+\alpha\\
1&1+\alpha^2+\alpha^3&\alpha^3\\
1&\alpha(1+\alpha+\alpha^2)&1+\alpha^2\\
1&1+\alpha+\alpha^2+\alpha^3&\alpha^3
\end{pmatrix}
%}_{H^T}
\end{small}
=0.
\end{equation*}
\hrulefill
\end{figure*}

\begin{example}
\label{exGPD}
Consider a $(16,7,6)$ EBCH\ code generated by
\setlength{\arraycolsep}{1pt}
 $$G=\left(
\begin{array}{cccccccc|cccccccc}
1&0&0&1&0&1&1&0&1&0&0&1&0&1&1&0\\
0&1&0&1&0&1&0&1&0&1&0&1&0&1&0&1\\
0&0&1&1&0&0&1&1&0&0&1&1&0&0&1&1\\
0&0&0&0&1&1&1&1&0&0&0&0&1&1&1&1\\\hline
0&0&1&0&1&0&1&1&0&0&0&1&1&0&0&0\\
1&0&0&0&0&0&1&0&1&1&0&1&1&0&0&0\\
1&1&1&1&1&1&1&1&0&0&0&0&0&0&0&0
\end{array}\right).
$$
Its GPD is given by  $G_1=\begin{pmatrix}
0& 0& 1& 1& 0& 0& 1& 1\\
0& 1& 0& 1& 1& 0& 1& 0\\
1& 1& 1& 1& 1& 1& 1& 1
\end{pmatrix}, G_2=\begin{pmatrix}
1&0&0&1&0&1&1&0\\
0&1&0&1&0&1&0&1\\
0&0&1&1&0&0&1&1\\
0&0&0&0&1&1&1&1
\end{pmatrix}, G_3=\begin{pmatrix}
0&0&0&1&1&0&0&0\\
1&1&0&1&1&0&0&0
\end{pmatrix}
$. 
%Observe that $G_1$ generates a code which is equivalent to a concatenation
%of  $(4,3,2)$ and $(2,1,2)$ codes, and $G_2$ generates a $(8,4,4)$ code.
\end{example}

GPD enables one to perform hard-decision decoding of code $\mathcal C$ as follows. Let $\mathcal C_i$ be the code generated by $G_i$. Consider a noisy codeword  $(y'|y'')=(c'|c'')+(e'|e'')$, where $e=(e'|e'')$ is an error vector. Compute $z=y'+y''=(c'+c'')+(e'+e'')$. One can decode
$z$ in  $\mathcal C_1$ to identify information vector $u'$ and codeword $c'+c''=u'G_1$. If this step is completed successfully, one can compute $\tilde y'=y'-u'(G_1+\tilde IG_3)$ and $\tilde y''=y''-u'\tilde I G_3$, and try to decode these vectors in $\mathcal C_2$. 
This algorithm can be easily tailored to implement soft-decision decoding.

One can see from \eqref{mGenPlotkin} that $\mathcal C_2$ has minimum distance $d_2\geq d/2$. However,  $d_1$ can be very low. Hence, the above described  algorithm may fail to correct even $t\leq \floor{(d-1)/2}$ errors. A workaround for this problem is to employ list decoding for $\mathcal C_1$ to identify a number of possible vectors $u'$,  for each of them decode the corresponding vectors $\tilde y',\tilde y''$ in $\mathcal C_2$, and select the codeword $(c'|c'')$ closest to the received sequence.

GPD may be also applied recursively. This results
in codes of length $1$  and dimension at most $1$, as discussed below. 

\section{Dynamic Frozen Symbols}
\label{sDFS}

\subsection{Representation of a linear code for SC decoding}

Consider an $(n=l^m,k,d)$ code $\mathcal C$ over $\F_q$ with check matrix $H$. Let $A=B_{l,m}F_l^{\otimes m}$ be a  matrix of an $n\times n$ polarizing transformation. Since $A$ is invertible, any vector of length $n$ can be obtained as an output $c_0^{n-1}=u_0^{n-1}A$ of the polarizing transformation. Let us investigate the constraints which need to be imposed on $u_0^{n-1}$, so that the output of the polarizing transformation is a codeword of $\mathcal C$.

These constraints are given by the equation 
 $u_0^{n-1}AH^T=0$.
By applying Gaussian elimination, one can construct the {\em constraint matrix} $V=QHA^T$, where $Q$ is an invertible matrix, such that all rows of $V$ end in distinct columns, i.e. the values $j_i=\max \set{t|V_{i,t}\neq 0}, 0\leq
i<n-k$ are distinct. It can be assumed without loss of generality that $V_{i,j_i}=-1$. Let $\mF=\set{j_i|0\leq
i<n-k}$. Then one obtains \begin{equation}
\label{mDynFrozen}
u_{j_i}=\sum_{s=0}^{j_i-1}u_sV_{i,s}, 0\leq i<n-k.
\end{equation}
These equations  can be considered as a generalization of the concept of frozen symbols, i.e. constraints of the form $u_{j_i}=0, j_i\in \mF$, used in the construction of polar codes. Observe that symbols $u_{j_i}, j_i\in \mF$ can take arbitrary values, which, however, depend on the values of some other symbols with smaller indices. Therefore, symbols $u_{j_i}$ given by \eqref{mDynFrozen} will be referred to as dynamic frozen symbols.

\begin{example}
\label{ex1676}
Consider $(16,7,6)$ EBCH code $\mathcal C$ over $\F_2$. The generator polynomial of the corresponding non-extended code has roots $\alpha, \alpha^3$ and their conjugates, where $\alpha$ is a primitive root 
of $x^4+x^3+1$.
The constraints on 
vector $u_0^{15}$, such that $u_0^{15}A\in \mathcal C$, 
are given by the equation at the top of this page. 
%\left(
%\setlength{\arraycolsep}{1pt}
%\begin{array}{cccccccccccccccc}
%\end{array}\right)\end{tiny}=0
%1+a+a^2+a^3&a+a^2+a^3&1+a^2+a^3&a^2+a^3&1+a+a^3&a+a^3&1+a^3&a^3&1+a+a^2&a+a^2&1+a^2&a^2&1+a&a&1&0\\ 
%a^2+a^3&a^3&a+a^3&a^3&a^2+a^3&1+a+a^2+a^3&1+a+a^2+a^3&a+a^3)&1&1&a+a^3&a^2+a^3&1+a+a^2+a^3&a^3&1&0\\
%1&0&0&0&0&0&0&0&0&0&0&0&0&0&0&0
Multiplying matrices, expanding their elements in the standard basis and 
applying elementary linear operations, one obtains
\begin{equation*}
u_0^{15}
\left(\setlength{\arraycolsep}{1pt}
\begin{array}{cccccccccccccccc}
0&0&0&0&0&0&0&0&0&0&1&0&1&0&0&0\\
0&0&0&1&0&1&0&0&0&0&1&0&0&0&0&0\\
0&0&0&0&0&1&0&0&0&1&0&0&0&0&0&0\\
0&0&0&0&0&0&0&0&1&0&0&0&0&0&0&0\\
0&0&0&1&0&0&1&0&0&0&0&0&0&0&0&0\\
0&0&0&0&1&0&0&0&0&0&0&0&0&0&0&0\\
0&0&1&0&0&0&0&0&0&0&0&0&0&0&0&0\\
0&1&0&0&0&0&0&0&0&0&0&0&0&0&0&0\\
1&0&0&0&0&0&0&0&0&0&0&0&0&0&0&0
\end{array}
\right)^T=0
\end{equation*}
This means that $u_0=u_1=u_2=u_4=u_8=0$, and 
$u_6=u_3$, $u_9=u_5$, $u_{10}=u_3+u_5$, $u_{12}=u_{10}$. Symbols
$u_3,u_5, u_7,u_{11},u_{13},u_{14}, u_{15}$ are
non-frozen.
\end{example}

Observe that any linear code of length $l^m$   can be represented by a system of equations \eqref{mDynFrozen}. This enables one to employ the SC decoding algorithm and its variations for decoding of arbitrary linear codes of length $l^m$. That is,
one can successively make decisions
\begin{equation}
\label{mSCDecisionRuleDF}
\widehat{u}_i=\begin{cases} \arg \max_{u_i\in\F_q} W_{n}^{(i)}(u_0^i|y_0^{n-1}),& i\not \in \mathcal F \\\sum_{s=0}^{i-1}u_sV_{t_i,s},&\text{otherwise},\end{cases}
\end{equation}
where $t_i$ is an integer, such that $j_{t_i}=i$. Observe that if $u_0^{i-1}$
are the correct values of the input symbols of the polarizing transformation,
the probability of symbol error $P_i$ in this case remains the same as in the case of classical polar codes. Hence, the error probability of the
considered code under
SC decoding can be still computed via \eqref{mSCErrProb}.

The set  $\mF$ of dynamic frozen symbol indices for a generic linear code is not guaranteed to contain all symbols with high error probability.  Hence, for most linear codes the SC decoding error probability,  given by \eqref{mSCErrProb}, far exceeds the error probability of other decoding algorithms. Substantially better performance can be obtained by employing list or stack SC decoding techniques. However, the list size (i.e. the decoding complexity) needed to obtain near-ML\ performance, in general, increases exponentially with code dimension.

The complexity of computing $W_{n}^{(i)}(u_0^i|y_0^{n-1})$ is exactly the
same as in the case of classical polar codes, i.e. $O(n\log n)$. However, evaluation of the
expression \eqref{mDynFrozen} may increase the decoding complexity to $O(n^2)$.

\subsection{Extended BCH\ codes and Arikan kernel}
Let us investigate in more details the structure of the set of dynamic frozen symbol indices of binary extended primitive narrow-sense BCH (EBCH) codes for the case of Arikan polarizing transformation. Observe that in this case construction of the system of equations \eqref{mDynFrozen} can be viewed as a recursive application of the GPD to the considered code.

It was shown in \cite{kolesnik1968cyclic,kasami1968new,delsarte1970generalized} that 
a punctured RM code of order $r$ and length $2^m$ is equivalent to a  cyclic code 
 with generator polynomial  $g(x)$ having roots $\alpha^i: 1\leq 
\wt(i)<m-r, 1\leq i\leq 2^m-2$, where $\alpha$ is a primitive element of $\F_{2^m}$, and $\wt(i)$ is the number of non-zero digits in the binary expansion of integer $i$.  Furthermore, it was shown in \cite{zinoviev1986dual} that an EBCH code $\mathcal C'$ of length $2^m$ with design distance $d\geq \delta(r,m)+3$ is a subcode of the RM code of order $m-r-1$, where $$\delta(r,m)=\max_{i:\wt(i)=r}\min\set{i2^j\bmod (2^m-1)|0\leq j<m}.$$
A recursive expression for $\delta(r,m)$ is derived in \cite{zinoviev1986dual}. One can consider a RM code of order $m-r-1$ as a polar code with the set of frozen symbol indices $\mF''=\set{i|\wt(i)\leq r}$. Hence,  the set of dynamic frozen symbol indices $\mF'$ for the EBCH code includes $\mF''$. It can be seen from \eqref{mArikanBhat1}--\eqref{mArikanBhat2} that $Z_{n,i}=O(Z_{1,0}^{2^{\wt(i)}})$. Hence, the set of  frozen symbols for EBCH codes includes all those ones, such that their Bhattacharyya parameters (and error probability $P_i\leq \frac{1}{2}Z_{n,i}$) decrease  slowly while decreasing the Bhattacharyya parameter $Z_{1,0}$ of the original channel. Most of these symbols have high error probability $P_i$. 

The above statement is true only if one employs standard digit ordering. That is, each coordinate $c_i$ of a codeword $(c_0,\dots,c_{n-1})$ can be associated with some $x_i\in \F_{2^m}$, so that all $x_i$ are distinct, and all codewords satisfy check equations
\begin{equation}
\label{mCheckEq}
\sum_{i=0}^{n-1}c_ix_i^j=0, 0\leq j<d-1.
\end{equation}
The standard digit ordering is given by $x_i=\sum_{j=0}^{m-1}X_{i,j}\beta_j$, where $i=\sum_{j=0}^{m-1}X_{i,j}2^j, X_{i,j}\in \set{0,1}$, and $\beta_0,\dots,\beta_{m-1}$ is some basis of $\F_{2^m}$. In what follows, more detailed characterization of the set of dynamic frozen symbols for EBCH\ codes will be derived.

Let $$C_t=\set{t2^i\bmod 2^{m-1}|0\leq i<m_t, t2^{m_t}\equiv t\bmod 2^m-1}$$ be a cyclotomic coset generated by $t$. Let $\mathcal Q$ be the set of minimal
cyclotomic coset representatives. It can be seen that all elements of a cyclotomic coset have the same weight. Therefore $$\sum_{\substack{s\in \mathcal Q\\\wt(s)=r}}m_s={m\choose r}.$$

\begin{theorem}
\label{tFrozenWeight}
Consider a $(2^m,k,d)$ extended primitive  narrow-sense BCH code over $\F_2$. Let $S=\set{i\in \mathcal Q|0\leq i<d-1}$. Let $N_t$ be the number of dynamic frozen symbols  $u_i$ for this code, such that $\wt(i)=t$. Then $N_t=\sum_{s\in S_t}m_s$, where $S_t=\set{s\in S|\wt(s)=t}$, and $m_s$ is the size of the cyclotomic coset generated by $s$.
\end{theorem}
\begin{proof}
Consider  parity check equation \eqref{mCheckEq}.  Let $x_i=\sum_{s=0}^{m-1}X_{i,s}\beta_s, X_{i,s}\in\set{0,1}$. Then 
\begin{align*}
x_i^j=\left(\sum_{s=0}^{m-1}X_{i,s}\beta_s\right)^j=\left(\sum_{s=0}^{m-1}X_{i,s}\beta_s\right)^{\sum_{t=0}^{m-1}j_t2^t}\nonumber\\=
\prod_{t=0}^{m-1}\left(\sum_{s=0}^{m-1}X_{i,s}\beta_s^{2^t}\right)^{j_t}
=\sum_{\substack{\wt(e_0^{m-1})\leq\wt(j)\\e_s\in \set{0,1}}} v_{j;e_0^{m-1}}\prod_{s=0}^{m-1}X_{is}^{e_s},
\end{align*}
where $ v_{j;e_0^{m-1}}\in \F_{2^{m_j}}$ are some coefficients.
Hence, any codeword $c_0^{n-1}$ satisfies 
\begin{equation*}
0=\sum_{\substack{\wt(e_0^{m-1})\leq\wt(j)\\e_s\in \set{0,1}}} v_{j;e_0^{m-1}}\sum_{i=0}^{n-1}c_i\prod_{s=0}^{m-1}X_{is}^{e_s},  j\in \tilde S,
\end{equation*}
where $\tilde S=\set{j2^l|j\in S,0\leq l<m_j,j2^{m_j}\equiv j\bmod 2^m-1}$.

It can be seen that the $i$-th row of $A=B_{2,m}F_2^{\otimes m}$ is a sequence of values of various monomials $X^{(a)}=\prod_{s=0}^{m-1}X_s^{a_{m-1-s}}, a_s\in\set{0,1}$ at  point $(X_{i,0},\dots,X_{i,m-1})\in \F_2^m$. Hence, $u_{e'}=\sum_{i=0}^{n-1}c_i \prod_{s=0}^{m-1}X_{i,s}^{e_s}$ is the value of the $e'$-th element of the input vector of the polarizing transformation, where $e'=\sum_{s=0}^{m-1}e_{s}2^{m-1-s}$,\ $e_s\in \set{0,1}$, and $u_{0}^{n-1}=c_0^{n-1}A$, so that 
$$0=\sum_{\substack{\wt(e_0^{m-1})\leq\wt(j)}} v_{j;e_0^{m-1}}u_{e'},  j\in \tilde S $$ 
Any such equation gives rise to $m_j$ equations with coefficients in $\F_2$.
Observe that there are  $\widetilde M_t=\sum_{i=t}^\rho \sum_{s\in S_i}m_s$  equations, which involve symbols $u_{e'}: t\leq \wt(e')\leq \rho$, where $\rho=\max_{0\leq j<d-1}\wt(j)$.  Hence, the number $\widehat M_t=\sum_{i=t}^\rho N_t$ of dynamic frozen symbols $u_i: t\leq \wt(i)\leq \rho,$ is upper bounded by $\widetilde M_t$. It can be also seen that $\widehat M_0=\widetilde M_0$. 

The equality $N_0=1=m_{0}$ holds for any EBCH\ code with $d\geq 2$. Assume that $N_t=\sum_{s\in S_t}m_s$ for all $t<t_0$, so that $\widehat M_{t_0}=\widetilde M_{t_0}$. Since $\widehat M_{t_0+1}=\widehat M_{t_0}-N_{t_0}\leq \widetilde M_{t_0+1}=\widetilde M_{t_0}-\sum_{s\in S_{t_0}}m_s$, one obtains $N_{t_0}\geq \sum_{s\in S_{t_0}}m_s$. Assume that this inequality is strict. 

Any codeword of the considered EBCH code can be represented as a vector of values of polynomial 
\begin{equation}
\label{mTracePoly}
f(x)=\sum_{t\in \mathcal Q\setminus S}\Tr_{m_t}(\gamma_tx^{n-1-t})
\end{equation}
 in distinct points $x_i\in \F_{2^m}$, where  $\gamma_i\in \F_{2^{m_i}}$, $\Tr_m(x)=\sum_{i=0}^{m-1}x^{2^i}$.  
This polynomial can be represented in multivariate form as 
\begin{align}
\label{mBCHMultivariate}
f(X_{0},\dots,X_{m-1})=&\sum_{t\in \mathcal Q\setminus S}\Tr_{m_t}\left(\gamma_t\prod_{j=0}^{m-1}\left(\sum_{l=0}^{m-1}\beta_l^{2^j}X_l\right)^{1-t_j}\right)\nonumber\\
=&\sum_{e_0,\dots,e_{m-1}\in\set{0,1}}u_{e'}\prod_{l=0}^{m-1}X_l^{1-e_l}
\end{align}
where $t=\sum_{j=0}^{m-1}t_j2^j, t_j\in\set{0,1}$, so that $c_i=f(x_i)=f(X_{i,0},\dots,X_{i,m-1})$.
Observe that the $e'$-th row of matrix $A$ can be considered as a table of values of $\prod_{l=0}^{m-1}X_l^{1-e_l}$ in various points of $\F_2^m$.
 Hence, $u_{e'}$ can be considered as input symbols of the polarizing transformation.  

Hence, the set of polynomials $f(X_0,\dots,X_{m-1})$ corresponding to the considered code contains $K_{t_0}={m\choose t_0}-\sum_{s\in S_{t_0}}m_s$ linearly independent polynomials given by \eqref{mBCHMultivariate} of degree $m-t_0$. Observe also that the forms of degree $m-t_0$ of these polynomials are also linearly independent. However, this is not possible since, by assumption, there are $N_t>\sum_{s\in S_{t_0}}m_s$ constraints on the coefficients of these forms. The obtained contradiction proves the theorem.
\end{proof}

The particular  set of dynamic frozen symbol indices $\mF'$  of the EBCH\ code depends on the basis being used. One may enumerate different bases of $\F_{2^m}$ and select the one which minimizes the  SC decoding error probability \eqref{mSCErrProb}. Similar approach was used in \cite{vardy1994maximumlikelihood,kasami1993optimum} to obtain trellis diagrams of EBCH codes.

\begin{figure}
\begin{center}
\includegraphics[width=0.5\textwidth]{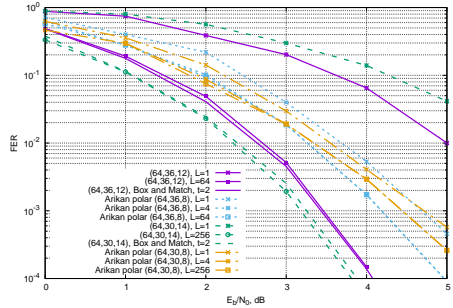}
\end{center}
\caption{Performance of list/stack SC decoding of extended BCH\ codes  }
\label{fDynFrozPerf}
\end{figure}

Theorem \ref{tFrozenWeight} and the existence of a RM supercode  suggest that the  SC  algorithm and its variations may work  for EBCH codes. Unfortunately, experiments show that this is true only for small $n$. Figure \ref{fDynFrozPerf} illustrates the performance of EBCH codes under list/stack SC algorithm with list size $L$ and  box-and-match  \cite{Valembois2004box} algorithm
with reprocessing order $t$, as well as Arikan polar codes. It can be seen that Arikan polar codes far outperform extended BCH codes in the case of list size $L$ equal to $1$ (i.e. classical SC decoding). However, higher minimum distance results in significant performance gain of EBCH codes under box-and-match near-ML\ decoding algorithm. Huge list size $L$ is needed in order to obtain comparable performance under list/stack SC decoding, while
Arikan polar codes achieve the near-ML performance already for  $L=4$.

\section{Polar subcodes}
\label{sPolarSubcodes}

It is possible  to show that the minimum distance of polar codes with Arikan kernel is given by $O(\sqrt{n})$ \cite{hussami2009performance}. This results in quite poor ML decoding performance. 

Exact performance analysis of the list/stack SC decoding algorithm, which
is commonly used to implement near-ML decoding of polar codes, still remains an open problem. It was empirically observed that in the low-SNR region codes with lower SC decoding error probability provide lower error probability under  list SC decoding. However, in the high-SNR region the performance of list/stack SC decoding algorithm depends mainly on code minimum distance. Therefore, we propose to explicitly construct codes with a given minimum distance, which would minimize the SC decoding error probability.

\begin{definition}
Consider  a $q$-ary input memoryless output symmetric channel $W(y|c)$
and an $(n=l^{m},k',d)$ code $\mathcal C'$ over $\F_q$, called parent code. Let $\mF'$ be the set of dynamic frozen symbol indices of $\mathcal C'$ for the case
of kernel $F_l$. An
$(n,k\leq k',\geq d)$ polar subcode $\mathcal C$ of code $\mathcal C'$ is defined as the set of vectors $c_0^{n-1}=u_0^{n-1}B_{l,m}F_l^{\otimes m}$, where $u_0^{n-1}$ simultaneously satisfies the dynamic freezing equations \eqref{mDynFrozen} for code $\mathcal C'$, and additional constraints $u_{s}=0$ for $k'-k$ indices $s\notin \mF'$ with the highest error probabilities $P_{s}$ for a given channel $W(y|c)$.
\end{definition}
Encoding of polar subcodes can be performed as 
\begin{equation}
\label{mPrecoding}
c=xWA,
\end{equation}
where $x$ is an information vector, $W$ is a matrix, such that $WV^T=0$, and $V$ is the dynamic freezing
constraint matrix. This can be considered as pre-coding
the data with some outer code with generator matrix $W$, and encoding its
codeword with a polar code. However, we do not explicitly specify an outer
code for this construction. Instead, we require that the obtained
codeword $c$ should belong to the parent code with sufficiently high minimum
distance. 

Polar codes with CRC \cite{tal2011list} and LDPC outer codes
 \cite{guo2014enhanced} can be considered as
a special case of the proposed construction. However, these codes employ
ad-hoc constraints \eqref{mDynFrozen}. Therefore, it is difficult to
control their minimum distance.

It must be recognized that the SC decoding error probability $P$ given by \eqref{mSCErrProb} of a polar subcode cannot be less than the SC decoding error probability of a classical polar code of the same length and dimension, constructed for the same channel using the same kernel $F_l$. Therefore, polar subcodes provide no advantage with respect to classical polar codes if SC decoding is used. However, significant performance gain under list/stack SC decoding can be obtained. 
Experiments show that for given values of $n,k,d$ polar subcodes with lower $P$ provide lower list/stack SC decoding error probability.  Hence, one should select $\mathcal C'$ so that its set $\mF'$   includes as many as possible indices $j_i$ corresponding to symbols with high error probability $P_{j_i}$.

\subsection{Arikan kernel}
\label{sArikaneBCHSubcode}
\subsubsection{The construction}
We propose to employ  EBCH\ codes as parent ones in the proposed construction
of polar subcodes.
Theorem \ref{tFrozenWeight} implies that the indices of the most of the frozen symbols of EBCH codes have low weight. Bounds \eqref{mArikanBhat1}--\eqref{mArikanBhat2}
imply that the Bhattacharyya parameter of the $i$-th bit subchannel is given
by $Z_{n,i}=O(Z_{1,0}^{2^{\wt(i)}})$. Hence,  employing EBCH codes as parent
ones in the proposed construction enables one to avoid freezing of bit subchannels
with low $Z_{n,i}$. This improves the performance of the obtained code under
SC decoding and its variations.

Observe that increasing minimum distance of the parent code causes more bit
subchannels with low $Z_{n,i}$ to be frozen. In order to keep code dimension
$k$ fixed, one needs to unfreeze some bits subchannels with high   $Z_{n,i}$.
This results in higher SC decoding error probability. This can be compensated
to a certain extent by employing list SC decoding and its variations with
larger list size. Unfortunately,
there are still no analytical techniques for finding a trade-off between
the performance and decoding complexity. We have to use simulations in order
to find  optimal values of the code minimum distance.
\begin{example}
Let us construct a $(16,6,6)$ polar subcode  of $(16,7,6)$ EBCH code considered in Example \ref{ex1676}, by optimizing it for the case of the
binary erasure channel with erasure probability $Z_{1,0}=0.5$. The vector of bit subchannel Bhattacharyya parameters (i.e. symbol erasure probabilities) equals $Z_{16}=(
\underline{0.999}, \underline{0.992}, \underline{0.985}, 0.77,  
\underline{0.96}$, $\underline{ 0.65}, 0.53, 0.1,             
\underline{0.9},   \underline{0.47},  \underline{0.35}$,  $3.7\cdot 10^{-2}, \underline{0.23},            1.5\cdot 10^{-2}, 
7.8\cdot 10^{-3}, 1.5\cdot 10^{-5})$.  Here the values 
corresponding to dynamic frozen symbols of the EBCH code are underlined. 
It can be seen that $u_3$ has the highest erasure probability $0.77$ among
not yet frozen symbols. Therefore, we propose to introduce an additional
constraint $u_3=0$. This is equivalent to removing the
first row from matrices $G_1$ and $G_3$ presented in Example \ref{exGPD}.
\end{example}
\begin{example}
\label{ex1024512}
Consider construction of a $(1024,512)$ code.  There exists a 
$(1024,513,116)$ EBCH code, which cannot  be decoded 
efficiently with (list) SC decoder. On the other hand, the classical polar code 
optimized for AWGN channel with $E_b/N_0=2 dB$ has minimum distance 16.  
One can take a $(1024,893,28)$ EBCH parent code $\mathcal C': RM(5,10)\subset \mathcal C'\subset RM(8,10)$ and freeze $381$ additional bit 
subchannels to obtain a $(1024,512,\geq 28)$ polar subcode with dynamic 
frozen symbols. The specification of the obtained code includes only $f=20$ non-trivial
equations \eqref{mDynFrozen} with $T=111$ terms, so the cost of evaluation of
dynamic frozen symbols is negligible compared to the cost of multiplication
by matrix $A$.
\end{example}
Observe that the SC decoding error probability $P$ of a $(n,k,d)$ polar subcode of any code cannot be less than the SC decoding error probability $\overline P$  for a classical $(n,k)$ polar code constructed for the same channel using
the same kernel. However, the performance of a polar subcode under list/stack SC decoding with sufficiently large list size $L$ may be substantially better.  It was empirically observed that the size of the list $L$ needed to obtain such gain increases with $P     $.
Hence, one needs to quantify the value of $\overline P/P$ needed to obtain a given minimum distance $d$. However, it would be easier to characterize the rate of a polar subcode with a given minimum distance, such that it has the same SC decoding error probability as a given classical polar code of the same length.

 Let  $\overline {\mathcal C}$ be a polar code with kernel $F_2$ of rate $\overline \rho(z)$,
 such that all symbols with  $Z_{n,i}<z$ are not frozen.    Consider now an $(n,k,d)$  polar subcode $\mathcal C$ of rate $\rho(z,d)=k/n$, obtained from  a $(n=2^{m},k'=\beta(m,d)n,d)$ EBCH code  $\mathcal C'$ by freezing all symbols $u_i$ with $Z_{n,i}\geq z$. The set of non-frozen symbols of code  $\mathcal C$ can be represented as  $\Delta(d,z)=\cup_{r=0}^{m-1}\left(\Delta(r,m,z)\setminus \mathcal F_{d,r}'\right)$, where $\Delta(r,m,z)=\set{i|0\leq i<2^m,\wt(i)=r,Z_{2^m,i}<z}$,
and $\mathcal F_{d,r}'$ is the set of dynamic frozen symbol indices $u_e$ of $\mathcal C'$, such that $\wt(e)=r$.  

It is quite difficult to find $|\Delta(r,m,z)|$ analytically, although it can be computed in polynomial time for any specific binary input output symmetric memoryless channel and values $r,m$ \cite{tal2011how}. Therefore, we propose to approximate it by employing an asymptotic expression for the fraction of common non-frozen symbols of a RM code of rate $R(m-r,m)=2^{-m}\sum_{j=0}^{m-r}{m\choose j}$ and a polar code  of length $n=2^m$ and rate $\rho$.
This value was shown in   \cite{hassani2013ratedependent} to converge with $m\rightarrow \infty$ to $\phi(\rho,r,m)=C\min\left(\frac{\rho}{C},R(m-r,m)\right)$, where $C$ is the capacity of the considered channel. Hence, the 
$2^{-m}|\Delta(r,m,z)|\approx \phi(\overline \rho(z) ,r,m)-\phi(\overline \rho(z) ,r+1,m)$. 

Therefore, one obtains 
\begin{equation}
\label{mMaxDynFrozen}
|\Delta(r,m,z)\setminus \mathcal F_{d,r}'|\geq \max(0,|\Delta(r,m,z)|-|\mathcal F_{d,r}'|),
\end{equation}
 so that 
\begin{align*}
\rho(z,d)\geq 2^{-m}\sum_{r=0}^m \max(0,|\Delta(r,m,z)|-|\mathcal F_{d,r}'|)\nonumber\\
\approx \sum_{r=0}^m\max\left(0,\phi(\overline \rho(z) ,r,m)-\phi(\overline \rho(z) ,r+1,m)-N_r2^{-m}\right),
\label{mRateDegr}
\end{align*}
where $N_r$ is given by Theorem \ref{tFrozenWeight}.

\begin{figure}
\includegraphics[width=0.5\textwidth]{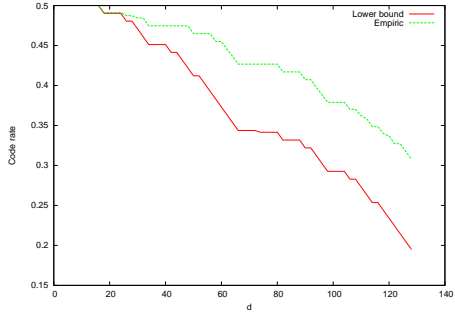}
\caption{Rate of polar subcodes of length 1024}
\label{fRateDist}
\end{figure}
Figure \ref{fRateDist} illustrates this bound together with the actual rate of $(n,k,d)$ polar subcodes of EBCH codes of length $n=1024$. The dimension $k$ of these subcodes was selected so that they achieve approximately the same successive cancellation decoding error probability at $E_s/N_0=-1$ dB as the classical Arikan polar code $(1024,512,16)$ constructed for the same value of $E_s/N_0$. It can be seen that the bound is quite loose. This is both due to \eqref{mMaxDynFrozen}, which assumes that all dynamic frozen symbols induced by the EBCH code correspond to subchannels with the lowest possible Bhattacharyya parameters, and application of an asymptotic expression for approximation of $|\Delta(r,m,z)|$.

It can be seen that the degradation of polar subcode rate with respect to a classical polar code is negligible for $d$ up to twice the minimum distance of the original polar code. 
\subsubsection{ Encoding and Decoding Complexity}
Encoding of the proposed polar subcodes of binary EBCH\ codes can be performed
via \eqref{mPrecoding} with complexity $C_W+\frac{1}{2}n\log n$, where $C_W=T-f$
is the cost of multiplication by matrix $W$, $T$ is the number  of terms
in the right-hand side of non-trivial equations \eqref{mDynFrozen}, and $f\leq
\sum_t N_t$
is the number  of such equations.
Systematic encoding  can be implemented using the approach introduced in \cite{sarkis2015flexible} with complexity $C_W+n\log n$. 

Theorem \ref{tFrozenWeight} implies that the set of dynamic frozen symbol
indices for a parent EBCH\ code includes only the ones with sufficiently small
weight. Furthermore, one can see that a dynamic freezing equation for symbol $u_i$ cannot involve symbols $u_j: \wt(j)>\wt(i)$. On the other hand, most
of the constraints $u_s=0$, imposed on bit subchannels with high error probability
$P_s$, correspond to low-weight integers $s$. This causes $f$ to be much less
than the value $\sum_t N_t$ predicted by Theorem \ref{tFrozenWeight} for the parent code,
so that matrix $W$ appears to be sparse.  This was illustrated in  Example
\ref{ex1024512}.  Hence, the encoding complexity of the proposed polar subcodes
does not exceed that of polar codes with $f$-bit CRC.

Decoding of polar subcodes can be performed using the same algorithms as
classical polar codes, which should be augmented with a subroutine for evaluation of dynamic frozen symbols. Hence, the number of operations with probabilities
or log-likelihood ratios remains the same as in the case of classical polar
codes. However, the cost of bit manipulations increases at least by $O(C_W)$.
 For example, in the case of Tal-Vardy list decoding algorithm and its derivatives, the values $u_j$, which are needed for evaluation of the dynamic frozen
 symbols, are not stored explicitly. One should either introduce for each
 path an additional
 array of size $f$, where the values of dynamic frozen symbols are accumulated,
 or recover $u_j$ from intermediate values. In the first case the decoding
 complexity  increases by $fLn+C_{W}L$ bit operations, where $L$ is the
 list size, since the additional
 arrays need to be copied while cloning the paths. In the second case the
 complexity depends on the specific structure of dynamic freezing constraints.

The sequential decoding algorithm \cite{miloslavskaya2014sequential} and
its block generalization \cite{trofimiuk2015block} were shown to be able
to decode polar codes with very low average complexity and good performance.
These algorithms can be naturally used in the case of polar subcodes.

\subsection{Improved polar subcodes with Arikan kernel}
Let us consider a $(n=2^m,k,d)$ polar subcode constructed as described in Section \ref{sArikaneBCHSubcode}. It can be represented as an IGCC with outer
codes of length $2^s, s<m$. It appears that most outer codes obtained
in this way are classical Arikan polar codes with quite low minimum distance
and high decoding error probability. Therefore we propose to employ the approach
suggested in \cite{trifonov2011generalized}. Namely, we impose the requirement
on outer codes to be $(2^s,k_i,d_i)$ EBCH\ codes (or their subcodes). The parameters $k_i, d_i$ are selected in order to minimize the MSD error probability, which is given by
$$P=1-\prod_{i=0}^{2^{m-s}-1}(1-\pi_i)$$
 under the constraint $\sum_{i=0}^{2^{m-s}-1}k_i=k$.  Here $\pi_i$ denotes
 the decoding error probability of the code utilized at the $i$-th level
 of the IGCC.
These probabilities
can be estimated, for example, using the tangential sphere bound \cite{poltyrev1994bounds}
together with density evolution \cite{tal2011how} or Gaussian approximation \cite{trifonov2012efficient}.

The obtained IGCC can be also represented
via a system of equations \eqref{mDynFrozen}.  The corresponding matrix $V$
is given by $$V=\left(\begin{array}{cccc}
\multicolumn{4}{c}{V'}\\
V_0&0&\dots&0\\
0&V_1&\dots&0\\
\vdots&\vdots&\ddots&\vdots\\

0&0&\dots&V_{2^{m-s}-1}\\
\end{array}\right),
$$
where $V'$ is a constraint matrix for a parent $(2^m,K\geq k,d)$ EBCH\ code, and
$V_i$ are constraint matrices for outer $(2^s,k_i,d_i)$ codes.

The codes obtained in this way are supposed to be decoded by the block sequential
decoding algorithm \cite{trofimiuk2015block} with block size at least $2^s$. This algorithm employs the fast tree-trellis list Viterbi algorithm
 \cite{roder2006fast} for decoding of outer codes of the IGCC. However, more efficient decoding techniques can be designed
 for specific outer codes.
  
The proposed approach can be considered as a generalization of the  construction suggested in \cite{mondelli2014from}. The Mondelli-Hassani-Urbanke codes
can be considered as GCC with inner Arikan codes and outer RM or
polar codes. Since EBCH\ codes provide higher minimum distance, one may
expect the improved polar subcodes to provide better performance.
\subsection{Polar subcodes with EBCH\ kernel}
For $l=2^\mu$ an extended BCH kernel can be obtained as matrix $F_l$, where  $((F_l)_{i+1,1},\dots,(F_l)_{i+1,l-1})$ is a vector of coefficients of $x^jg_{i'}(x),$ where  $g_{i'}(x)$ is a generator polynomial of a $(l-1,l-1-i')$ BCH code, and $j$ is the smallest non-negative integer, such that $i=j+i'$.   Furthermore, one has $(F_l)_{0,0}=1$ and $(F_l)_{i+1,0}=\sum_{j=1}^{l-1}(F_l)_{i+1,j}$. Figure \ref{fKernel32} presents an example of the EBCH\ kernel.
\begin{figure}
$$\begin{pmatrix}\begin{smallmatrix}
1& 0& 0& 0& 0& 0& 0& 0& 0& 0& 0& 0& 0& 0& 0& 0& 0& 0& 0& 0& 0& 0& 0& 0& 0& 0& 0& 0& 0& 0& 0& 0\\ 
1& 1& 0& 0& 0& 0& 0& 0& 0& 0& 0& 0& 0& 0& 0& 0& 0& 0& 0& 0& 0& 0& 0& 0& 0& 0& 0& 0& 0& 0& 0& 0\\ 
1& 0& 1& 0& 0& 0& 0& 0& 0& 0& 0& 0& 0& 0& 0& 0& 0& 0& 0& 0& 0& 0& 0& 0& 0& 0& 0& 0& 0& 0& 0& 0\\ 
1& 0& 0& 1& 0& 0& 0& 0& 0& 0& 0& 0& 0& 0& 0& 0& 0& 0& 0& 0& 0& 0& 0& 0& 0& 0& 0& 0& 0& 0& 0& 0\\ 
1& 0& 0& 0& 1& 0& 0& 0& 0& 0& 0& 0& 0& 0& 0& 0& 0& 0& 0& 0& 0& 0& 0& 0& 0& 0& 0& 0& 0& 0& 0& 0\\ 
1& 0& 0& 0& 0& 1& 0& 0& 0& 0& 0& 0& 0& 0& 0& 0& 0& 0& 0& 0& 0& 0& 0& 0& 0& 0& 0& 0& 0& 0& 0& 0\\ 
1& 1& 0& 1& 0& 0& 1& 0& 0& 0& 0& 0& 0& 0& 0& 0& 0& 0& 0& 0& 0& 0& 0& 0& 0& 0& 0& 0& 0& 0& 0& 0\\ 
1& 0& 1& 0& 1& 0& 0& 1& 0& 0& 0& 0& 0& 0& 0& 0& 0& 0& 0& 0& 0& 0& 0& 0& 0& 0& 0& 0& 0& 0& 0& 0\\ 
1& 0& 0& 1& 0& 1& 0& 0& 1& 0& 0& 0& 0& 0& 0& 0& 0& 0& 0& 0& 0& 0& 0& 0& 0& 0& 0& 0& 0& 0& 0& 0\\ 
1& 0& 0& 0& 1& 0& 1& 0& 0& 1& 0& 0& 0& 0& 0& 0& 0& 0& 0& 0& 0& 0& 0& 0& 0& 0& 0& 0& 0& 0& 0& 0\\ 
1& 0& 0& 0& 0& 1& 0& 1& 0& 0& 1& 0& 0& 0& 0& 0& 0& 0& 0& 0& 0& 0& 0& 0& 0& 0& 0& 0& 0& 0& 0& 0\\ 
1& 1& 0& 0& 1& 0& 1& 1& 0& 1& 1& 1& 0& 0& 0& 0& 0& 0& 0& 0& 0& 0& 0& 0& 0& 0& 0& 0& 0& 0& 0& 0\\ 
1& 0& 1& 0& 0& 1& 0& 1& 1& 0& 1& 1& 1& 0& 0& 0& 0& 0& 0& 0& 0& 0& 0& 0& 0& 0& 0& 0& 0& 0& 0& 0\\ 
1& 0& 0& 1& 0& 0& 1& 0& 1& 1& 0& 1& 1& 1& 0& 0& 0& 0& 0& 0& 0& 0& 0& 0& 0& 0& 0& 0& 0& 0& 0& 0\\ 
1& 0& 0& 0& 1& 0& 0& 1& 0& 1& 1& 0& 1& 1& 1& 0& 0& 0& 0& 0& 0& 0& 0& 0& 0& 0& 0& 0& 0& 0& 0& 0\\ 
1& 0& 0& 0& 0& 1& 0& 0& 1& 0& 1& 1& 0& 1& 1& 1& 0& 0& 0& 0& 0& 0& 0& 0& 0& 0& 0& 0& 0& 0& 0& 0\\ 
1& 1& 1& 1& 1& 0& 1& 0& 1& 1& 1& 1& 1& 0& 0& 0& 1& 0& 0& 0& 0& 0& 0& 0& 0& 0& 0& 0& 0& 0& 0& 0\\ 
1& 0& 1& 1& 1& 1& 0& 1& 0& 1& 1& 1& 1& 1& 0& 0& 0& 1& 0& 0& 0& 0& 0& 0& 0& 0& 0& 0& 0& 0& 0& 0\\ 
1& 0& 0& 1& 1& 1& 1& 0& 1& 0& 1& 1& 1& 1& 1& 0& 0& 0& 1& 0& 0& 0& 0& 0& 0& 0& 0& 0& 0& 0& 0& 0\\ 
1& 0& 0& 0& 1& 1& 1& 1& 0& 1& 0& 1& 1& 1& 1& 1& 0& 0& 0& 1& 0& 0& 0& 0& 0& 0& 0& 0& 0& 0& 0& 0\\ 
1& 0& 0& 0& 0& 1& 1& 1& 1& 0& 1& 0& 1& 1& 1& 1& 1& 0& 0& 0& 1& 0& 0& 0& 0& 0& 0& 0& 0& 0& 0& 0\\ 
1& 1& 0& 1& 0& 1& 0& 1& 1& 0& 1& 1& 0& 0& 1& 0& 0& 0& 1& 1& 0& 1& 0& 0& 0& 0& 0& 0& 0& 0& 0& 0\\ 
1& 0& 1& 0& 1& 0& 1& 0& 1& 1& 0& 1& 1& 0& 0& 1& 0& 0& 0& 1& 1& 0& 1& 0& 0& 0& 0& 0& 0& 0& 0& 0\\ 
1& 0& 0& 1& 0& 1& 0& 1& 0& 1& 1& 0& 1& 1& 0& 0& 1& 0& 0& 0& 1& 1& 0& 1& 0& 0& 0& 0& 0& 0& 0& 0\\ 
1& 0& 0& 0& 1& 0& 1& 0& 1& 0& 1& 1& 0& 1& 1& 0& 0& 1& 0& 0& 0& 1& 1& 0& 1& 0& 0& 0& 0& 0& 0& 0\\ 
1& 0& 0& 0& 0& 1& 0& 1& 0& 1& 0& 1& 1& 0& 1& 1& 0& 0& 1& 0& 0& 0& 1& 1& 0& 1& 0& 0& 0& 0& 0& 0\\ 
1& 1& 1& 1& 0& 0& 1& 0& 0& 0& 1& 0& 1& 0& 1& 1& 1& 1& 0& 1& 1& 0& 1& 0& 0& 1& 1& 0& 0& 0& 0& 0\\ 
1& 0& 1& 1& 1& 0& 0& 1& 0& 0& 0& 1& 0& 1& 0& 1& 1& 1& 1& 0& 1& 1& 0& 1& 0& 0& 1& 1& 0& 0& 0& 0\\ 
1& 0& 0& 1& 1& 1& 0& 0& 1& 0& 0& 0& 1& 0& 1& 0& 1& 1& 1& 1& 0& 1& 1& 0& 1& 0& 0& 1& 1& 0& 0& 0\\ 
1& 0& 0& 0& 1& 1& 1& 0& 0& 1& 0& 0& 0& 1& 0& 1& 0& 1& 1& 1& 1& 0& 1& 1& 0& 1& 0& 0& 1& 1& 0& 0\\ 
1& 0& 0& 0& 0& 1& 1& 1& 0& 0& 1& 0& 0& 0& 1& 0& 1& 0& 1& 1& 1& 1& 0& 1& 1& 0& 1& 0& 0& 1& 1& 0\\ 
1& 1& 1& 1& 1& 1& 1& 1& 1& 1& 1& 1& 1& 1& 1& 1& 1& 1& 1& 1& 1& 1& 1& 1& 1& 1& 1& 1& 1& 1& 1& 1
\end{smallmatrix}\end{pmatrix}$$
\caption{Extended BCH kernel $F_{32}$}
\label{fKernel32}
\end{figure}

Let $C(q,s,m,J)$ be a code over $\F_q$ of length $q^{sm}$, which consists of evaluation vectors of $m$-variate polynomials $\frac{1}{(q^s-1)^m}A(x_0,\dots,x_{m-1})$, such that all their coefficients $A_{j_0,\dots,j_{m-1}}$ are equal to zero for $(j_0,\dots,j_{m-1})\in J$, in various points $(x_0,\dots,x_{m-1})\in \F_{q^s}^m$.
Let $L(h)=\set{(j_0,\dots,j_{m-1})|0\leq j_i<q^s, \sum_i{j_i}>h}$, and $L^*(h)=\set{j=\sum_{i=0}^{m-1}j_iq^{si}|(j_0,\dots,j_{m-1})\in L(h)}$. It was shown in \cite{liesenfeld1993equivalence} that codes $ C(q,s,m,L(h))$ and $C(q,sm,1,L^*(h))$  are equivalent. That is, these codes can be considered as  GCC with inner and outer extended cyclic codes. Furthermore, outer codes in this construction can be recursively decomposed in the same way.
This implies that an EBCH\ code has a supercode, which is equivalent to
a GCC with inner and outer extended cyclic codes.
This is similar to the fact of existence of a RM supercode for any
EBCH code, which was used above to show that EBCH\ codes have ``not so bad'' set of dynamic frozen symbol indices in the case of Arikan polarizing transformation.

Therefore, we propose the following code construction. Let $\mathcal C'$ be an $(2^{\mu m},k',d)$ EBCH code, such that its $t$-th locator is $x_t=\sum_{j=0}^{m-1}\beta_{t_j}\gamma_j$, where $(\gamma_0,\dots,\gamma_{m-1})$ is a basis of $\F_{2^{\mu m}}$ considered as a vector space over $\F_{2^\mu}$,  $t=\sum_{j=0}^{m-1}t_j2^{\mu j}, 0\leq t_j<2^{\mu}$, and $\beta_i$ is the $i$-th element of $\F_{2^\mu}$.   The above described construction of EBCH\ kernel corresponds to the case $\beta_0=0,\beta_i=\alpha^{i-1}, 1\leq i<2^\mu$, where $\alpha$ is a primitive element of $\F_{2^\mu}$.  Then
one can construct  polar subcode of $\mathcal C'$ using the polarizing transformation $A=B_{l,m}F_l^{\otimes
m}$.  

The proposed construction requires one to be able to compute symbol error probabilities
$P_i$. To the best of our knowledge, there are still no analytical techniques
for solving this problem, except for the case of the  binary erasure channel \cite{miloslavskaya2012design}.
Therefore, we use simulations to obtain these values.

The obtained polar subcode can be considered as an instance of the IGCC introduced
in Section \ref{sIGCC}. Indeed, let us consider a subset $\mF_s=\set{j_i\in
\mF|\floor{j_i/l}=s}$ of the set of dynamic frozen symbol indices corresponding
to the $s$-th block, $0\leq s<l^{m-1}$, and let $V^{(s)}$ be the corresponding
$|\mF_s|\times l^m$ submatrix of $V$. It can be assumed without loss of generality
 that $V^{(s)}$ has an identity submatrix in columns with indices in
$\mF_s$, so that $V^{(s)}=(\Delta_s|\underbrace{(\Sigma_s|I)\Pi_s}_{|\mF_s|\times
l}|0)$,
where $\Pi_s$ is a $l\times l$ permutation matrix, and $\Delta_s, \Sigma_s$
are some matrices.  Therefore, one obtains a system of equations 
$$u_{sl^{m-1}}^{sl^{m-1}+l-1}((\Sigma_s|I)\Pi_s)^T=u_0^{sl^{m-1}-1}\Delta_s^T.$$
Its solution is given by $$u_{sl^{m-1}}^{sl^{m-1}+l-1}=v(I|\Sigma_s)\Pi_s+
u_0^{sl^{m-1}-1}\Delta_s^T(0|I)\Pi_s,$$
where $v$ is an arbitrary vector in $\F_2^{l-|\mF_s|}$.  Hence, instead of
successive decoding of symbols $u_{sl^{m-1}},\dots,u_{sl^{m-1}+l-1}$ according
to \eqref{mSCDecisionRuleDF}, one can recover them jointly by decoding in a coset $x_s+\mathbb C_s$, where $\mathbb C_s$ is a code generated by matrix $(I|\Sigma_s)\Pi_sF_l$, and $x_s=
u_0^{sl^{m-1}-1}\Delta_s^T(0|I)\Pi_sF_l$, as shown in Figure \ref{fMSD}.
This enables one  to improve the performance and/or reduce the decoding
complexity.

\begin{figure*}
\subfigure[$L=32$  \label{f1024Arikan32}]{
\includegraphics[width=0.495\textwidth]{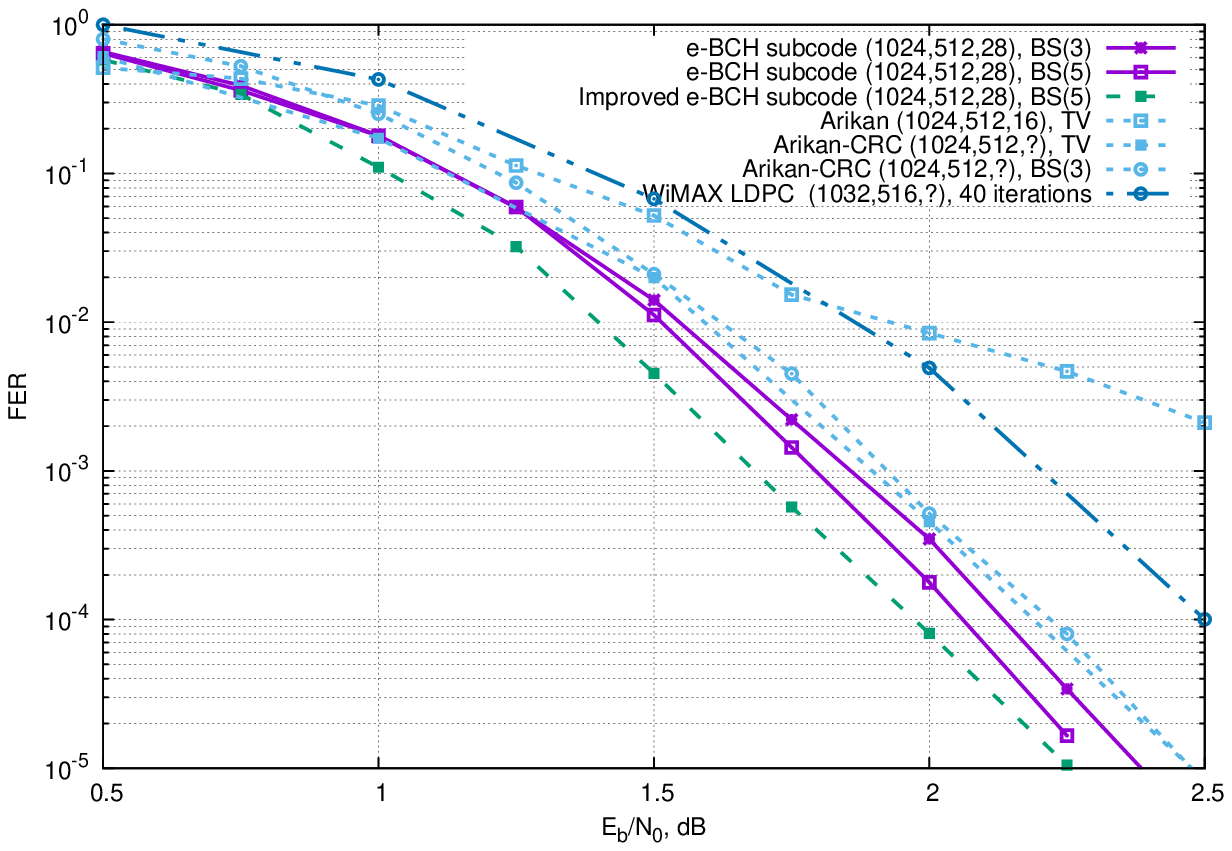}
}\subfigure[$L=256$ ]{
\includegraphics[width=0.495\textwidth]{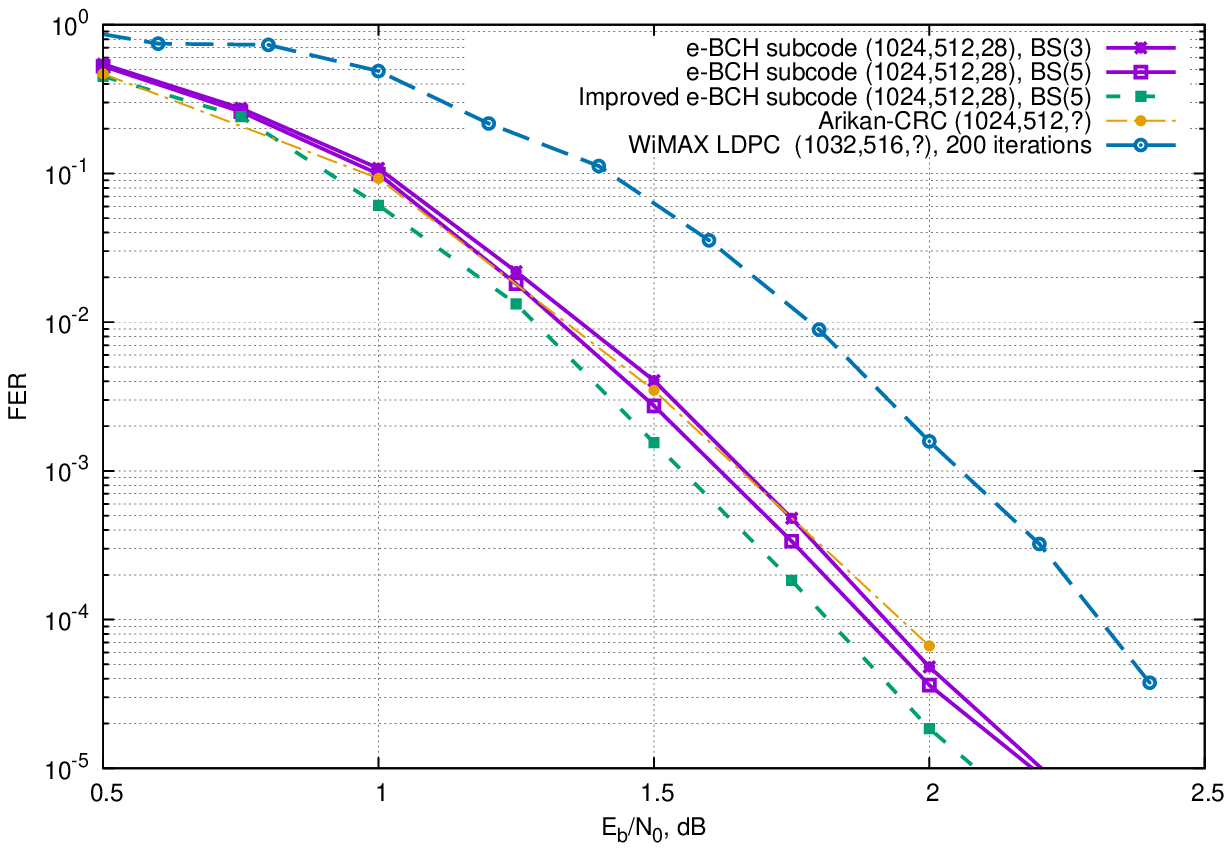}
}
\caption{Performance of polar subcodes with Arikan kernel of length 1024
   }
\label{f1024Arikan}
\end{figure*}

\subsection{Reed-Solomon kernel}
The results of \cite{liesenfeld1993equivalence} allow us to extend the proposed
construction of polar subcodes of EBCH\ codes to the case of Reed-Solomon
(RS)
kernel over $\F_q$. The RS kernel  is given by matrix $F_l$, where $(F_l)_{i,j}=\beta_j^{l-1-i}$,
and $\beta_j$ are some distinct elements of $\F_q, l\leq q$. 
 It was shown in \cite{mori2014source} that for $l\leq
q$ the Reed-Solomon kernel provides the highest possible polarization rate.
However, polar codes with RS kernel still suffer from low minimum distance.

In order to obtain a code with better performance, one can set $l=q$ and represent an  $(n=q^m,k',d)$ EBCH\ code of length $q^m$ over $\F_q$, such that its $t$-th locator is $x_t=\sum_{j=0}^{m-1}\beta_{t_j}\gamma_j$, where $(\gamma_0,\dots,\gamma_{m-1})$ is a basis of $\F_{q^{ m}}$ considered as a vector space over $\F_{q}$,  $t=\sum_{j=0}^{m-1}t_jq^{ j}, 0\leq t_j<q^{m}$,  via a
system of equations \eqref{mDynFrozen}, and introduce additionally $k'-k$
static freezing constraints $u_i=0$ for non-frozen subchannels $W^{(i)}_n$
with the highest error probability.  Again,  simulations
have to be used for performance evaluation of bit subchannels.  The obtained codes can be decoded using the
techniques presented in \cite{trifonov2013polar,trifonov2014binary}.
\begin{example}
Consider construction of a $(16,8,6)$ polar subcode over $\F_4$. The $4\times
4$ Reed-Solomon kernel is given by $$F_4=\begin{pmatrix}
0& 1& 1& 1\\
0& 1& \beta+1& \beta\\
0& 1& \beta& \beta+1\\
1& 1& 1& 1
\end{pmatrix},$$
where $\beta$ is a primitive element of $\F_4$. The check matrix of the 
$(16,9,6)$ \setlength{\arraycolsep}{1.6pt}
parent EBCH\ code is $$H=\left(\begin{array}{cccccccccccccccc}
1&1&1&1&1&1&1&1&1&1&1&1&1&1&1&1\\
0&1&\beta&\beta^2&0&1&\beta&\beta^2&0&1&\beta&\beta^2&0&1&\beta&\beta^2\\
0&0&0&0&1&1&1&1&\beta&\beta&\beta&\beta&\beta^2&\beta^2&\beta^2&\beta^2\\
0&1&\beta^2&\beta&\beta&\beta^2&1&0&1&0&\beta&\beta^2&\beta^2&\beta&0&1\\
0&0&0&0&1&1&1&1&\beta^2&\beta^2&\beta^2&\beta^2&\beta&\beta&\beta&\beta\\
0&1&1&1&\beta&1&0&\beta&\beta&\beta&1&0&\beta&0&\beta&1\\
0&0&0&0&\beta^2&\beta^2&\beta&\beta&\beta^2&\beta&\beta^2&\beta&\beta^2&\beta&\beta&\beta^2
\end{array}\right)$$
This corresponds to the following constraint matrix for
the 
polarizing transformation $A=B_{4,2}F_4^{\otimes 2}$:  
$$V=\left(\begin{array}{cccccccccccccccc}
0&0&0&0&0&0&1&0&0&\beta^2&0&0&1&0&0&0\\
0&0&0&\beta^2&0&0&1&0&0&1&0&0&0&0&0&0\\
0&0&0&0&0&0&0&0&1&0&0&0&0&0&0&0\\
0&0&0&0&1&0&0&0&0&0&0&0&0&0&0&0\\
0&0&1&0&0&0&0&0&0&0&0&0&0&0&0&0\\
0&1&0&0&0&0&0&0&0&0&0&0&0&0&0&0\\
1&0&0&0&0&0&0&0&0&0&0&0&0&0&0&0
\end{array}\right).$$
In the case of transmission of a binary image of the output of the polarizing
transformation $A$ over the AWGN\ channel with $E_s/N_0=-1$ dB, the symbol error
probabilities were found to be $(\underline{0.74},  \underline{0.7},  \underline{0.55}, 0.27, \underline{0.58},  0.33,  0.12,  0.02,  \underline{0.23} , 0.04 , 4 \cdot 10^{-3}  , 
2\cdot 10^{-4} , \underline{0.03} ,4\cdot 10^{-4} , 3\cdot 10^{-6}   , <10^{-6})$. Hence,
we propose to set additionally $u_5=0$.
 \end{example}

\section{Numeric results}
\label{sNumeric}

In this section we present simulation results illustrating the performance
of proposed polar subcodes of EBCH codes in the case of AWGN channel and BPSK modulation.
For comparison, we present also the results for the
case of classical polar codes with the corresponding kernels, polar codes with Arikan kernel and CRC-16 (Arikan-CRC) \cite{tal2011list}, LTE turbo
code, as well as LDPC codes specified in WiMAX and CCSDS standards. For polar subcodes with Arikan kernel we have used the block sequential (BS(s)) \cite{trofimiuk2015block} decoding algorithm\footnote{For $s\leq 3$ the block sequential decoding algorithm
provides slightly inferior performance compared to the probability-domain implementation
of the Tal-Vardy list decoding algorithm with the same list size $L$, but requires
much smaller number of arithmetic operations.},
where $2^s$ is the length of outer codes in the IGCC representation of the
corresponding polar subcode.  For polar codes with the BCH\ kernel, the sequential
decoding algorithm \cite{miloslavskaya2014sequentialBCH} was used, which
is based on the order-statistics soft-input hard-output decoding of the component
codes. Both probability-domain
implementation of the Tal-Vardy
list decoding algorithm (TV) and the block sequential decoding algorithm were used for decoding of polar codes with Arikan kernel an CRC. 
Observe that in the case of polar codes with CRC the block sequential decoding algorithm provides slightly worse performance compared to the original Tal-Vardy
algorithm, but has much lower complexity. Belief propagation
algorithm with flooding schedule was used for decoding of LDPC codes.

\begin{figure}
\includegraphics[width=0.5\textwidth]{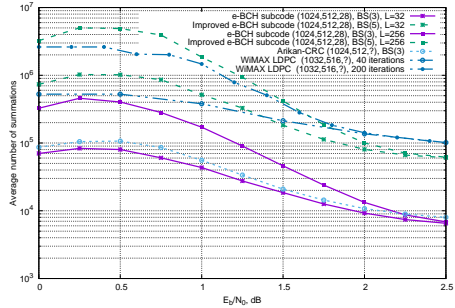}
\caption{Decoding complexity for polar subcodes with Arikan
kernel}
\label{fArikanComplexity}
\end{figure}

Figure \ref{f1024Arikan} illustrates the performance of codes\footnote{In order to ensure reproducibility of the results, we have set up a web site
\url{http://dcn.icc.spbstu.ru/index.php?id=polar} containing the specifications
of the considered polar subcodes with Arikan kernel.} of length $\approx
1024$. It can be seen that polar subcodes of EBCH codes provide significant performance
gain with respect to the classical polar codes of the same code length
and dimension. Furthermore, they outperform polar codes with Arikan kernel
and CRC. Observe that increasing $s$ in the case of the block sequential
decoding algorithm, i.e. employing in the decoder a representation of a polar
subcode as an IGCC with longer outer codes, results in better performance.
The best performance is achieved by improved polar subcodes, where outer
EBCH\ codes of length $32$ were selected so that the MSD
error probability of the corresponding IGCC is minimized.

Figure \ref{fArikanComplexity} illustrates the average number of summation
operations performed by the block sequential and belief propagation decoding
algorithms for the case of polar subcodes and LDPC codes, respectively. Observe
that decoding polar subcodes requires slightly lower average number of operations compared to polar codes with CRC, since the dynamic freezing constraints prevent
the sequential decoder from constructing wrong paths up to the final phase of decoding.

  It
can be also seen that for $s=3$ decoding of polar codes requires 10 times less
operations compared to LDPC codes. For $s=5$ the complexity becomes
comparable. Furthermore, the average number of operations for the case of
$L=256,s=3$ is less than in the case of $L=32,s=5$. From these results one
may conclude that it is more advantageous to increase $L$ instead of $s$.
However, increasing $s$ enables one to use the improved construction
of polar subcodes. We also believe that the block sequential decoding algorithm can be further simplified by employing more efficient decoding algorithms
for outer EBCH codes.

\begin{figure*}
\subfigure[rate $1/3 $]{\includegraphics[width=0.495\textwidth]{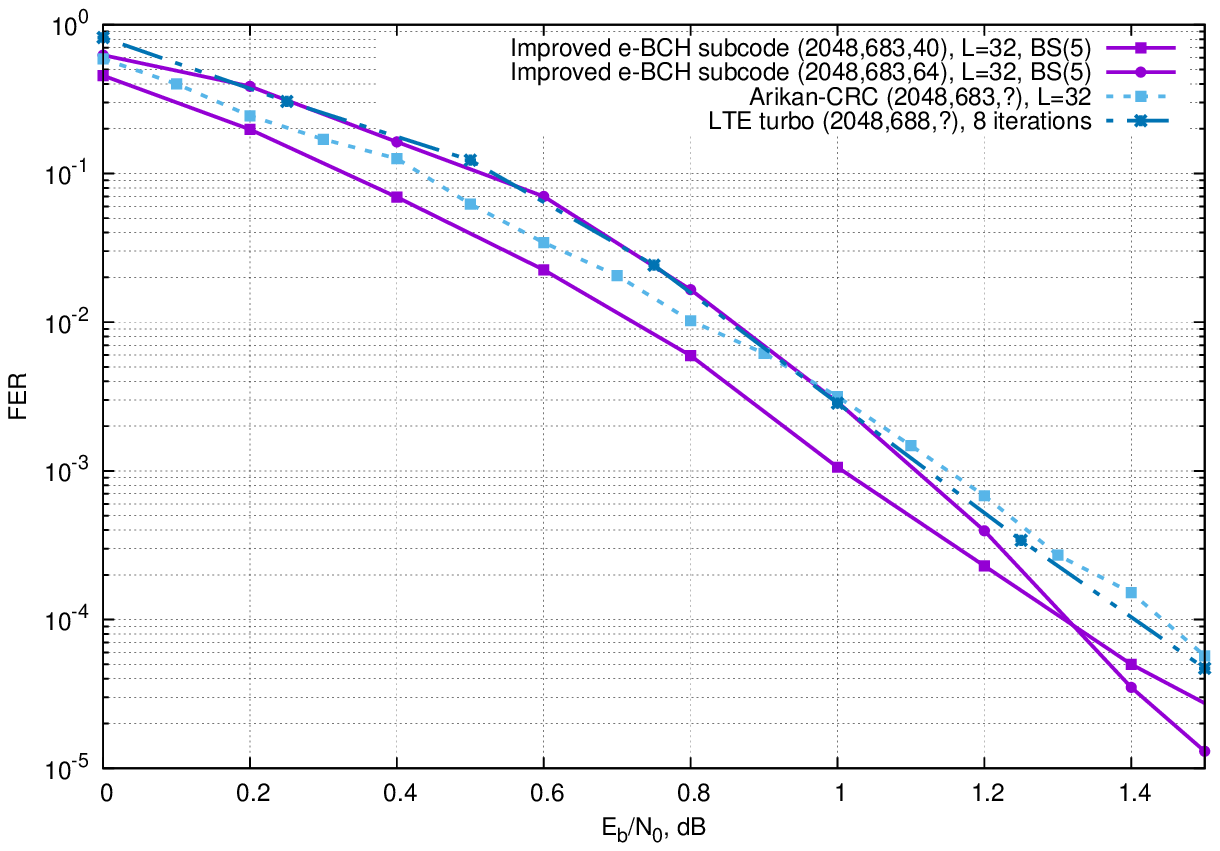}
}
\subfigure[rate $1/2  $ \label{f2048_12}]{\includegraphics[width=0.495\textwidth]{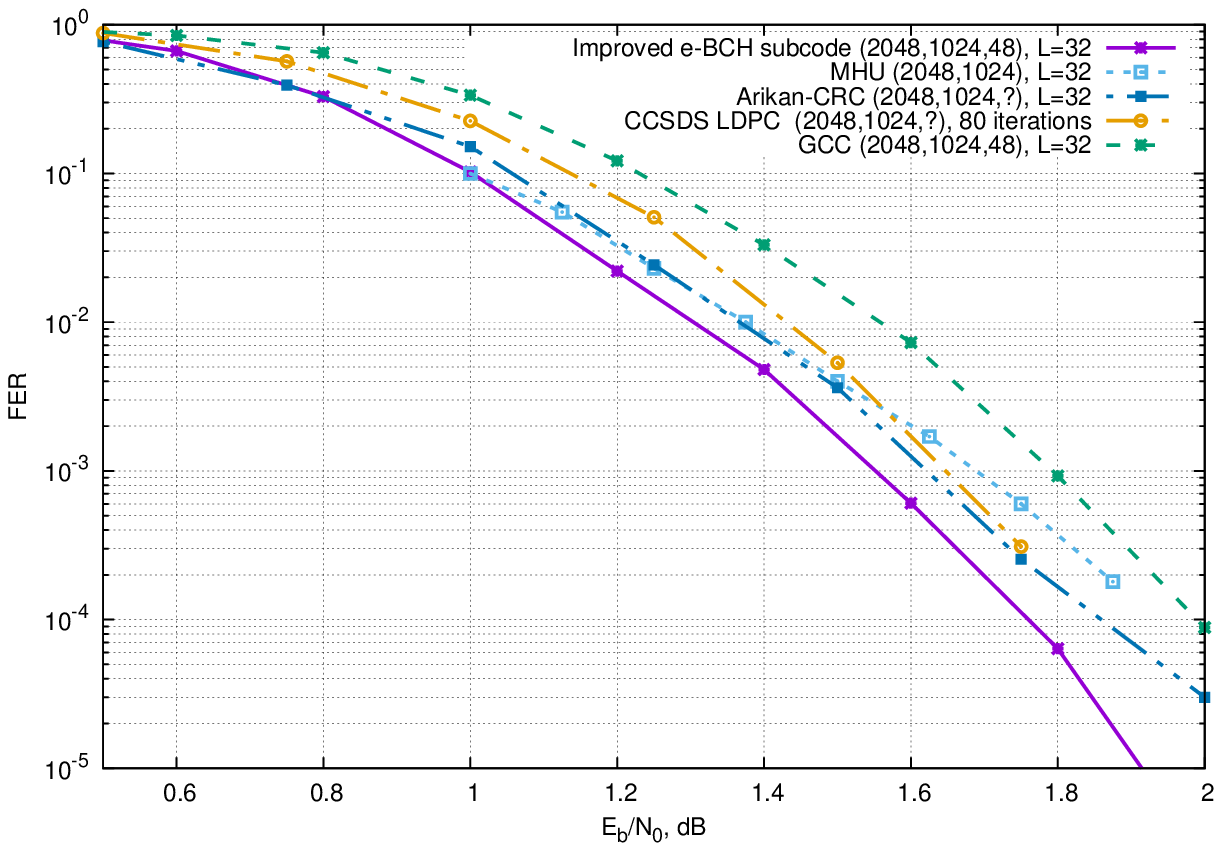}
}
 
\caption{Performance of codes of length 2048    }
\label{f2048}
\end{figure*}

Figure \ref{f2048} presents the performance of codes of length $2048$.
It can be seen that the proposed improved polar subcodes with Arikan kernel provide substantially
better performance compared to the case of LDPC and turbo codes. Observe
that increasing minimum distance of the polar subcode results in better performance
in the high SNR region. For comparison,
we provide also the results for the MHU construction, reproduced
from \cite{mondelli2014from}. As it may be expected, the improved polar subcode,
which employs EBCH\ outer codes, provides better performance than the MHU
code, which employs outer RM and
Arikan polar codes. For comparison, we report also results for the case of
a GCC with outer EBCH\ codes of length $32$, which was obtained as described in \cite{trifonov2011generalized}, and decoded with  the block sequential
algorithm. It appears that some of the outer codes of the IGCC corresponding
to the improved polar subcode, which correspond to good bit subchannels, have higher rate, while those corresponding to bad bit subchannels have lower
rate than in the case of the classical GCC optimized for the same SNR. This
causes the performance of the improved polar subcode to be better than that
of the classical GCC.

\begin{figure*}
\subfigure[EBCH kernel \label{f1024BCH} ]{\includegraphics[width=0.495\textwidth]{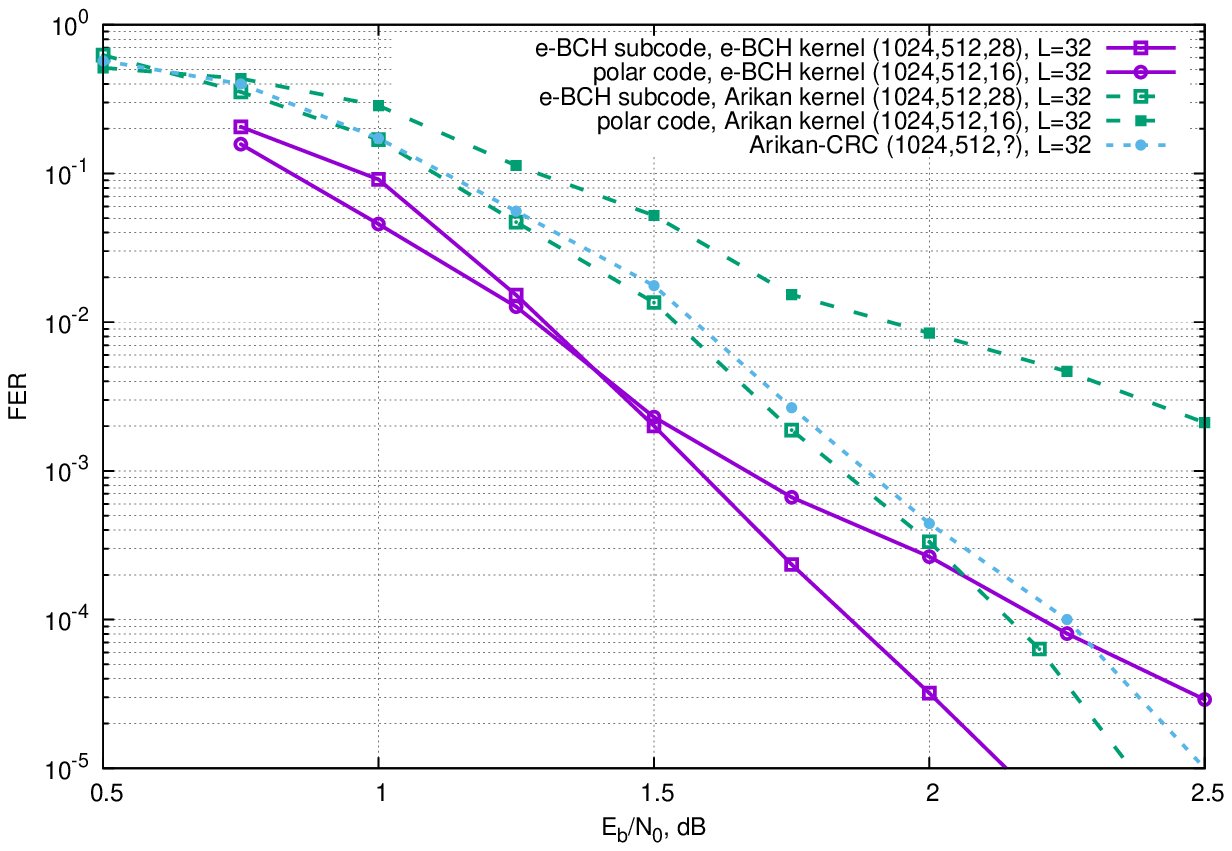}}
\subfigure[Reed-Solomon kernel \label{fRS} ]{\includegraphics[width=0.495\textwidth]{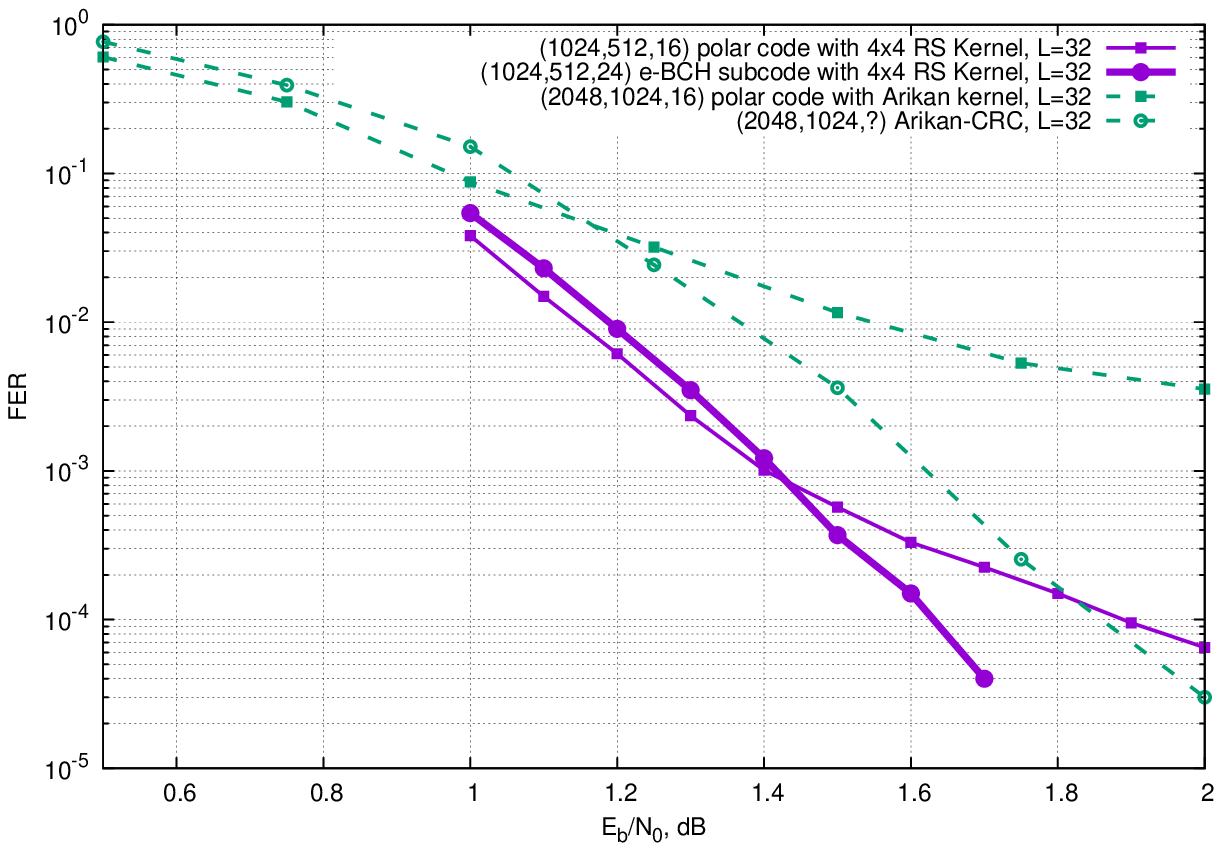}}
\caption{Performance of  codes with non-Arikan kernels}
\label{fNonArikan}
\end{figure*}

Figure \ref{f1024BCH} presents the performance of  codes with the EBCH\
kernel $F_{32}$. It can be seen that these codes outperform those with Arikan kernel. For the case of polar subcodes,  even better performance can be obtained by increasing the list size $L$ at the cost of higher decoding complexity. Figure \ref{fRS} presents the performance of the binary image of polar codes
and polar subcodes with $4\times 4$ Reed-Solomon kernel over $\F_{2^2}$. It can be seen that
classical polar codes with Reed-Solomon kernel have quite low minimum distance,
similarly to the case of Arikan kernel, but still provide better performance
compared to a polar code with Arikan kernel with comparable parameters. In both cases employing the proposed
construction of polar subcodes of EBCH\ codes results in  improved minimum distance and even better performance compared to the codes presented in Figures \ref{f1024Arikan32} and \ref{f2048_12}.

% \begin{figure}
% \includegraphics[width=0.5\textwidth]{plots/16384_8192.eps}
% \caption{Performance of codes of length 16384}
% \label{f16384}
% \end{figure}
%Similar results are observed for codes of length  $\approx 16384$, as shown
%in Figure \ref{f16384}.

\section{Conclusions}
In this paper the construction of polar subcodes of linear block codes was
introduced, which is based on the concept of dynamic frozen symbols. It enables one to obtain codes with higher minimum distance than
classical polar codes, which can still be efficiently decoded using the derivatives of the list successive cancellation algorithm.    
Although we do
not have a proof that the proposed codes achieve the channel capacity,
they  were shown to outperform some of the existing LDPC and turbo codes
of moderate lengths.
Many  existing constructions based on polar codes, such as polar codes with
CRC, can be considered
as a special case of the proposed polar subcodes. 

Unfortunately, due to lack of analytical techniques for predicting the performance
of list/stack SC decoding algorithms, heuristical methods were used in this paper  to construct the codes. Any progress in the performance analysis of
these algorithms may lead to design of better codes. Another way to improve
the performance of the proposed codes is to use longer outer EBCH codes.
This, however, requires development of efficient list soft decision decoding algorithms for them.

Furthermore, an extension of the concept of generalized concatenated codes
was provided, as well as a new method for representing linear block codes
in a form, which enables application of the SC 
algorithm and its variations for their decoding. This approach enables one
to construct polar subcodes with improved performance, as well as a more efficient decoding algorithm for them. It allows also 
near-ML decoding of short Reed-Solomon codes \cite{trifonov2014successiveEng,trifonov2014successive,miloslavskaya2014sequentialRS}.

\section*{Acknowledgements}
The authors thank R. Morozov for help in running simulations for polar
codes with Reed-Solomon kernel.

The authors thank the anonymous reviewers and the Guest Editor for their helpful comments, which
have greatly improved the quality of the paper.

\bibliographystyle{IEEETran}

\end{document}